\documentclass[journal,draftclsnofoot,onecolumn]{IEEEtran}
\usepackage{graphicx}
\usepackage{amsmath,amssymb}
\usepackage{dsfont}
\usepackage{color}
\usepackage{cases}
\usepackage{bm}
\usepackage{amsfonts}
\usepackage[font=footnotesize]{caption}
\usepackage{indentfirst}
\usepackage{cite}
\usepackage{bbm}
\usepackage{caption}
\usepackage{algorithm}
\usepackage{algpseudocode}
\usepackage{booktabs}
\usepackage{subfigure}
\usepackage{multirow}
\usepackage{amsthm}
\usepackage{setspace}

  \newcommand{\sys}{\mathsf}
\theoremstyle{remark}
\newtheorem{remark}{Remark}
\theoremstyle{}
\newtheorem{theorem}{Theorem}
\theoremstyle{}

\newtheorem{corollary}{Corollary}
\newtheorem{lemma}{Lemma}
\theoremstyle{}
\newtheorem{definition}{Definition}
\theoremstyle{remark}
\newtheorem{example}{Example}

\newtheorem{proposition}{Proposition}
\theoremstyle{definition}

\makeatletter
\newcommand{\tabcaption}{\def\@captype{table}\caption}
\makeatletter
 \allowdisplaybreaks[4]

\renewcommand{\bm}{\mathbf}

\ifCLASSINFOpdf
\else
\fi

\definecolor{newcolor}{rgb}{0.5,0,1}

\newcommand{\mw}[1]{{\color{black}#1}}
\newcommand{\qifa}[1]{{\color{black}#1}}

\newcommand{\sheng}[1]{{\color{black}#1}}

\newcommand{\maw}[1]{{\color{black}#1}}
\newcommand{\new}[1]{{\color{black}#1}}

\begin{document}

\title{A Fundamental Storage-Communication Tradeoff \mw{for} Distributed Computing with Straggling Nodes}

\author{Qifa Yan,
        Mich$\grave{\mbox{e}}$le Wigger,~Sheng Yang, and~Xiaohu Tang 
\thanks{Q.  Yan and M.  Wigger are with  LTCI, T$\acute{\mbox{e}}$l$\acute{\mbox{e}}$com Paris, IP Paris, 91120 Palaiseau, France. E-mails: qifay2014@163.com, michele.wigger@telecom-paristech.fr.}
\thanks{S. Yang is with  L2S,
(UMR CNRS 8506), CentraleSup$\acute{\mbox{e}}$lec-CNRS-Universit$\acute{\mbox{e}}$ Paris-Sud,
91192 Gif-sur-Yvette, France. Email: sheng.yang@centralesupelec.fr. }
\thanks{X. Tang is with  the Information Security and National Computing Grid Laboratory,
Southwest Jiaotong University, 611756, Chengdu, Sichuan, China.
Email:  xhutang@swjtu.edu.cn.}
\thanks{Part of this work has been presented in ISIT 2019 \cite{YanISIT2019}.
 The work of Q. Yan and M. Wigger has been supported by
the ERC under grant agreement 715111. The work of X. Tang was supported in part by the National Natural Science Foundation of China under Grant 61871331.}
}
\maketitle

\IEEEpeerreviewmaketitle
\begin{abstract}
Placement delivery arrays  for distributed
  computing~(Comp-PDAs) have recently been proposed as a framework to construct \mw{universal} computing schemes for
  MapReduce-like systems. In this work, we extend this concept to systems with straggling nodes, i.e., to systems where a subset of the nodes  cannot accomplish the assigned
  map computations in due time. Unlike most previous works
  that focused on computing linear
  functions, our results \mw{are universal and} apply for arbitrary map and reduce functions. Our contributions are as follows. Firstly, we show how to construct a \mw{universal}  \new{coded} computing scheme for MapReduce-like  systems with straggling nodes from any given Comp-PDA. We also  characterize the storage  and communication loads of the resulting scheme in
  terms of the Comp-PDA parameters.  \mw{Then, we prove an
  information-theoretic converse bound on the
  storage-communication (SC) tradeoff achieved by \mw{universal} computing schemes with straggling nodes. We  show that \mw{the information-theoretic bound}
  matches the performance achieved by the coded computing schemes with straggling nodes corresponding to the Maddah-Ali and Niesen (MAN) PDAs, i.e., to the Comp-PDAs describing Maddah-Ali and Niesen's coded caching scheme.} Interestingly, the same Comp-PDAs (the MAN-PDAs)  are optimal for any number of straggling nodes, which implies that the map phase of optimal coded computing schemes does not need to be adapted to the number of stragglers in the system.
  We finally  prove that while the points that lie exactly on the fundamental SC tradeoff cannot be achieved with Comp-PDAs that require smaller number of files than the MAN-PDAs,  this is possible for some of the points that lie close to the SC tradeoff. For these latter points,  the decrease in the requested number of files can be exponential in the number of nodes of the system.
\end{abstract}

\begin{IEEEkeywords}
  Distributed computing, storage, communication, straggler, MapReduce,
placement delivery array.
\end{IEEEkeywords}


\section{Introduction}

Distributed computing has emerged as one of the most important
paradigms to speed up large-scale data analysis tasks.
One of the most popular programming models is
MapReduce~\cite{MapReduce} which has been used to parallelize
computations across distributed computing nodes, e.g., for machine
learning tools~\cite{Liu2015,MachineLearing}.

Consider the task of computing $\sys{D}$ output functions from $\sys{N}$
files through $\sys{K}$ nodes. With MapReduce, each output function
$\phi_d$, \mw{for $1\leq d\leq \sys{D}$,} can be decomposed into
 \begin{itemize}
   \item $\sys{N}$ \emph{map functions} $f_{d,1},\ldots,f_{d,\sys{N}}$,
     each depending on exactly one different file; and
   \item a \emph{reduce function} $h_d$ that combines the outputs of the $\sys{N}$ map functions.
 \end{itemize}
 Each node $k$ is responsible for computing a subset of
 $\sys{\frac{D}{K}}$ output functions  through three phases. In the first \emph{map phase}, \mw{a central server stores a subset of files  $\mathcal{M}_k$ at node $k$, for each $k\in[\sys{K}]$. Each node $k$ then}   computes all the \mw{$\sys{D}$} \emph{intermediate values} (IVAs) $f_{1,n}(w_n),\ldots,f_{\sys{D},n}(w_n)$ \mw{corresponding to each of its}  stored files $w_n\in\mathcal{M}_k$. In the subsequent \emph{shuffle phase}, it creates a signal from its computed IVAs and sends  the signal to all the other nodes. \maw{Based on the received exchanged signals and the locally computed IVAs, in the final reduce phase it reconstructs all the IVAs pertaining to its own output functions and calculates the desired outputs.}

 Recently, Li \emph{et al.} proposed a scheme named coded distributed
 computing (CDC) to reduce the communication load for data shuffling
 between the map  and  reduce phases \cite{Li2018Tradeoff}. The idea is
 to create multicast opportunities by duplicating the files and
 computing the corresponding map functions at different nodes.
 It is shown that the CDC scheme achieves the fundamental
\emph{storage-communication tradeoff}, i.e., it has the lowest
communication load for a given storage constraint. This result has been extended in various directions.
For example,    \cite{Yan2018SCC,Fragouli,YanD3C} account also for the
computation load during the map phase;  \cite{Qian2017How}  studies the computation
 \mw{resource-allocation problem}; \cite{Li2017Framework,Li2017Egdge, Li2018Wireless,Lampiris2018ISIT} consider wireless (noisy)~networks between computation nodes;
\cite{Song2} considers \mw{a} model \mw{where during the shuffle phase each node  broadcasts only to a random subset of the nodes}.



 In this paper, we consider a setup where \mw{during the map phase} each node takes a random
amount of time to compute its desired map functions\cite{Lee2018codes}. In this case,
instead of waiting for all the nodes to finish the assigned
computations, which can cause an intolerable delay, data shuffling starts
as soon as any set of $\sys{Q}$ nodes, $\sys{Q}\in[\sys{K}]$, terminate \mw{their
map procedures}. \mw{The set of $\sys{Q}$ nodes that first terminate the map procedure}  are called \emph{active
nodes}, while the remaining $\sys{K-Q}$ nodes are called
\emph{straggling nodes} or \emph{stragglers}. Note that
\mw{the stragglers are not identified prior to the beginning of the map phase and} the map phase
has to be designed without such knowledge.


Distributed computing systems with straggling nodes  have mainly
been studied in the context of a server-worker framework.  \mw{In this framework, a central server distributes the raw data to the workers like in the above described map phase, but following this map phase the workers directly communicate their  intermediate results to the server, which then produces the final outputs. (Thus, under the server-worker framework, all final outputs   are  calculated at the server and not at the distributed computing nodes as is the case in MapReduce systems.)}  
\mw{Under the server-worker framework, distributed computing systems with straggling nodes have for example been studied in
\cite{Lee2018codes,Reisizadeh2017conf,QYu2017Polynomial,Qyu2018polyminal,LeeMatrix,HierarchicalCoding,Cadambe2018,Baharav2018ISIT,Kiani2018ISIT,Nuwan2018,secretsharing}, which focused on  high-dimensional matrix-by-matrix or matrix-by-vector
multiplications,  and in
\cite{Tandon2017Gradient,Raviv2017,Zachary2018ISIT,Halbawi2018ISIT,Deniz18},
which proposed  codes for gradient computing.}


 \mw{Fewer works studied  MapReduce systems with straggling nodes
(hereafter referred to as straggling systems) which are more relevant for the present article. Specifically, Li \emph{et al.}
\cite{Li2016unified} proposed to incorporate MDS codes into the CDC
scheme to cope with straggler nodes. Their construction however works only when the map functions  accomplish matrix-by-vector
multiplications. Improved constructions were proposed by Zhang \emph{et al.} {\qifa{by choosing the parameters of MDS code and CDC scheme separately in a more flexible way}}
\cite{Zhang2018Improved},  but also their techniques are applicable only for  map functions that are matrix-by-matrix
multiplications. In many practical applications such as
computations in neural networks and machine learning, the map functions
are  non-linear and can be very complicated with little structure.} This motivates us to investigate the MapReduce framework with straggling nodes for general map and reduce functions. \mw{In particular, we will present universal coded computing schemes that can be applied to arbitrary straggling systems, irrespective of the specific map and reduce functions. Moreover, we will show \new{the} optimality of our schemes  among the class of universal schemes that do not rely on special properties of the map and reduce functions.}

More specifically, in this work, we first propose a systematic
construction of \mw{universal} coded computing schemes for straggling systems from any
 \emph{placement delivery array for distributed computing
(Comp-PDA)} \cite{Yan2018SCC}. A placement delivery array (PDA) is an
array  whose entries are either a special symbol $``*"$ or some
integer numbers called ordinary symbols. It was introduced in \cite{Yan2017PDA} to represent
in a single array both the
placement and the delivery of  coded caching schemes
with uncoded prefetching. \qifa{In particular, the coded caching schemes proposed by Maddah-Ali and Niesen in \cite{Maddah2014Fundamental} \mw{can be represented as PDAs \cite{Yan2017PDA}. The corresponding PDAs will be referred to as MAN-PDAs, and, as we will see, they play a fundamental role also in coded computing with stragglers.}}
\qifa{PDAs have further been  generalized  to other coded caching scenarios such as device-to-device models  \cite{d2d2019wang}, combination networks \cite{Yan2018ISIT}, networks with private demands \cite{Subpacket2019Private}, and  medical data sharing problems \cite{Sun2020medical}. Moreover, several different PDA  constructions have been proposed in \cite{Cheng2019, Cheng2019variant, HyperGraph, Cheng2019group, bipartite}.}
\mw{In this paper our focus is on   a
subclass of PDAs, called Comp-PDAs, 
which
were} introduced in \cite{Yan2018SCC} to describe coded computing
schemes for MapReduce systems without straggling nodes. In this paper, we
show that Comp-PDAs  can also be used to construct coded computing schemes for
straggling systems, \mw{and we express the storage
and computation loads of the obtained schemes in terms of the Comp-PDA parameters.}

\mw{We then proceed to characterize  the fundamental
storage-communication (SC) tradeoff for straggling systems by showing
that the SC tradeoff curve achieved by  coded computing schemes
obtained  from the \qifa{MAN-PDAs} matches a new information-theoretic
converse for universal computing schemes with stragglers.
That means, our converse bounds the SC tradeoffs achieved by
coded computing schemes that apply to arbitrary map and reduce
functions. For special map and reduce functions, e.g., linear functions,
it is possible to find tailored coded computing schemes  that
achieve better SC tradeoffs, than \sheng{the one} implied by our information-theoretic converse, see e.g., \cite{Zhang2018Improved}. It is worth pointing out that the MAN-PDA based coded computing schemes adopt a fixed storage strategy irrespective of the active set size $\sys{Q}$. This implies that the fundamental SC-tradeoff curve remains unchanged even if the   active set size $\sys{Q}$ has not yet been determined during the map phase. \maw{The proposed schemes are thus optimal also  in a scenario where the system imposes a strict time constraint  proceeds  to the shuffle and reduce phases with  the random number of nodes that have by then terminated their IVA calculations.}}

 In a final part of the manuscript, we study the complexity of optimal (or near-optimal) coded computing schemes. In fact, a
major practical limitation of  the SC-optimal coded computing schemes based on MAN-PDAs  is that they  can only be
implemented if the number of files \mw{$N$ 
in the system grows
exponentially with the number of nodes $\sys{K}$.  However, as we show in this paper,
in most cases, MAN-PDAs achieve their corresponding fundamental SC pairs with smallest possible number of files, i.e., with smallest \emph{file complexity}, among all Comp-PDA
based coded computing schemes. The mentioned practical limitation   is thus not a weakness of the MAN-PDAs, but seems  inherent to \maw{PDA-based coded computing schemes for stragglers}. Interestingly, the problem can be circumvented by slightly backing off from the SC-optimal tradeoff curve. We show that the coded computing schemes corresponding to \maw{some of}  the  Comp-PDAs  in \cite{Yan2017PDA}  achieve SC pairs close to the fundamental SC tradeoff curve but with
significantly smaller number of files than the optimal MAN-PDAs.
\maw{More precisely, we fix an integer $q$, let the number of nodes $\sys{K}$ be a  multiple of  $q$, and the storage load $r$ be such that $\frac{r}{\sys{K}}\in\{\frac{1}{q},\frac{q-1}{q}\}$ holds.  We compare the Comp-PDA in \cite{Yan2017PDA} and the  MAN-PDA for such pairs $(\sys{K},r)$  while we let  both of them tend to infinity proportionally. This comparison shows that the
 ratio of the minimum required number of files of the Comp-PDA in \cite{Yan2017PDA} and the MAN-PDA vanishes as
$O\left(e^{\sys{K}(1-\frac{1}{q})\ln \frac{q}{q-1}}\right)$,  while the ratio of their communication loads  approaches~$1$.}}

\mw{We summarize the contributions of this paper:
\begin{enumerate}
  \item We establish a general framework for constructing universal coded computing schemes for  straggling systems from Comp-PDAs, and evaluate their SC pairs in terms of the Comp-PDA parameters.
  \item We derive the fundamental SC tradeoff for any universal straggling system by means of an  information theoretic converse that matches the SC pairs achieved by coded computing schemes with stragglers based on the MAN-PDAs.
  \item We prove that, while in most cases points on the fundamental SC tradeoff curve can be achieved only with the same file complexity as MAN-PDA based schemes, points close to the fundamental SC tradeoff curve can be achieved with significantly smaller file complexities. 
\end{enumerate}

}

The remainder of this paper is organized as follows. Section
\ref{sec:model} formally describes our model, and Section~\ref{sec:PDA}
reviews the definitions of PDAs and Comp-PDAs; Section \ref{sec:mainresults} presents the main results of this paper; \mw{Sections \ref{sec:CC-PDA} to \ref{sec:proofofthmF} present the major proofs of our results;} and Section \ref{sec:conclusion} concludes this paper.

\paragraph*{Notations} For positive integers $n,k$ such that $n\geq k$, we use the notations
$[n]\triangleq\{1,2,\ldots,n\}$, and
$[k:n]\triangleq\{k,k+1,\ldots,n\}$. The binomial coefficient is denoted
by $C_n^k\triangleq \frac{n!}{k!(n-k)!}$ for $n\ge k\ge0$; we set $C_n^k =
0$ when $k<0$ or $k>n$. For $k\leq n$, we use $\mathbf{\Omega}_n^k$ to denote the collection 
 of all subsets of $[n]$ of cardinality $k$, i.e., $\mathbf{\Omega}_n^k\triangleq\{\mathcal{T}\subseteq[n]:|\mathcal{T}|=k\}$.  The binary field is denoted by $\mathbb{F}_2$
and the $n$ dimensional vector space over $\mathbb{F}_2$ is denoted by
$\mathbb{F}_2^n$. We use $|\mathcal{A}|$ to denote the
cardinality of the set $\mathcal{A}$, while for a signal $X$, $|X|$ is
the number of bits in $X$. 
\sheng{The order of set operations is from left to right}. \mw{Finally, $\mathbbm{1}(\cdot)$ denotes the indicator function that evaluates to $1$ if the statement in the parenthesis is true and it evaluates to $0$ otherwise.}

\section{System Model}\label{sec:model}

A
$(\sys{K,Q})$ \emph{straggling system} is
parameterized by the positive integers
\begin{IEEEeqnarray}{c}
\sys{K,Q,N,D,U,V,W,}\notag
\end{IEEEeqnarray}
as described in the following.
The system aims  to compute $\sys{D}$ output functions $\phi_1,\ldots,\phi_\sys{D}$ through $\mathsf{K}$ distributed computing nodes from $\sys{N}$ files.
Each output function $\phi_d:\mathbb{F}_{2}^{\mathsf{NW}}\rightarrow
\mathbb{F}_{2}^\mathsf{U}$ ($d\in [\sys{D}]$) takes as inputs the length $\sys{W}$ files in the  library $\mathcal{W}=\{w_1\ldots,w_{\mathsf{N}}\}$,
 and outputs a bit stream of length $\mathsf{U}$, i.e.,
 \begin{IEEEeqnarray}{c}
 u_d=\phi_d(w_1,\ldots,w_\mathsf{N})\in\mathbb{F}_{2}^\mathsf{U}.\notag
 \end{IEEEeqnarray}

Assume that the computation of the output functions $\phi_d$ can be decomposed as:
\begin{IEEEeqnarray}{c}
\phi_d(w_1,\ldots,w_{\mathsf{N}})=h_d(f_{d,1}(w_1),\ldots,f_{d,\mathsf{N}}(w_{\mathsf{N}})),\notag
\end{IEEEeqnarray}
where
\begin{itemize}
  \item the \emph{map function}
$f_{d,n}: \mathbb{F}_2^\mathsf{W}\rightarrow \mathbb{F}_2^\mathsf{V}$
maps the file $w_n$ into a binary stream of length $\mathsf{V}$, called intermediate value~(IVA), i.e.,
 \begin{IEEEeqnarray}{c}
 v_{d,n}\triangleq f_{d,n}(w_n)\in \mathbb{F}_2^\mathsf{V},\quad
 \forall~n\in[\mathsf{N}];\notag
 \end{IEEEeqnarray}
  \item the \emph{reduce function}
  $h_d: \mathbb{F}_2^{\mathsf{NV}}\rightarrow \mathbb{F}_2^\mathsf{U}$, 
  maps the IVAs $\{v_{d,n}\}_{n=1}^{\mathsf{N}}$
 into the output stream
   \begin{IEEEeqnarray}{c}
 u_d=\phi_d(w_1,\ldots,w_{\sys{N}})=h_d(v_{d,1},\ldots,v_{d,{\mathsf{N}}}).\notag
   \end{IEEEeqnarray}
   \end{itemize}

 Notice that a decomposition into map and reduce functions is always possible. In fact, trivially, one can set
   the map and reduce functions to be the identity and output functions respectively, i.e., $f_{d,n}(w_n)=w_n,$ and $h_d=\phi_d,~\forall~n\in [\sys{N}],~d\in[\sys{D}]$, in which case $\sys{V}=\sys{W}$. However, to mitigate the communication cost during the shuffle phase, one would prefer a decomposition such that the length of the IVAs  is  as small as possible  while still allowing the nodes  to compute the final outputs. The computation is carried out  through three phases, namely, the
map, shuffle, and reduce phases.
\begin{enumerate}
  \item \textbf{Map Phase:} Each node $k\in[\sys{K}]$ stores a subset  of files $\mathcal{M}_{k}\subseteq \mathcal{W}$,  
 and tries to compute all the IVAs from the files in $\mathcal{M}_k$, denoted by $\mathcal{C}_k$:
\begin{IEEEeqnarray}{c}
\mathcal{C}_{k}\triangleq\{v_{d,n}:d\in[\mathsf{D}],w_n\in\mathcal{M}_{k}\}.\label{set:computedIVAs}
\end{IEEEeqnarray}
Each node has a random amount of time to compute its corresponding IVAs. To limit
latency of the system, the coded computing scheme proceeds with
the shuffle and reduce phases as soon as a fixed number of
$\mathsf{Q}\in[\mathsf{K}]$ nodes have terminated the map computations.
These nodes are called \emph{active nodes}, and the set of all active nodes is called \emph{active set}, whereas the other
$\sys{K}-\sys{Q}$ nodes are called \emph{straggling nodes}.
 For simplicity,
we consider the symmetric case in which each subset $\mathcal{Q}\subseteq [\sys{K}]$ of size $|\mathcal{Q}|=\mathsf{Q}$ is active with same probability.
Let the random variable $\mathbf{Q}$ denote the random active set.
Then,
\begin{IEEEeqnarray}{c}
  \Pr\left\{\mathbf{Q}= \mathcal{Q} \right\} =  \frac{1}{C_{\mathsf{K}}^{\mathsf{Q}}}, \quad \forall~ \mathcal{Q} \in \mathbf{\Omega}_\sys{K}^{\sys{Q}}.\notag
\end{IEEEeqnarray}
In our model, we also assume that the map phase has been designed in a way that all the files can be
recovered\footnote{In this paper, we thus exclude the ``outage'' event in
which some active sets cannot compute the given function due to missing
files.}\label{footnote1} from any active set of size $\sys{Q}$.
Hence, for any file $w_n\in\mathcal{W}$, the number of nodes  storing this file  \new{$t_n$} must
satisfy
\begin{IEEEeqnarray}{c}
t_n\geq\mathsf{K}-\mathsf{Q}+1,\quad\forall\; n\in[\mathsf{N}].
\label{eq:tn}
\end{IEEEeqnarray}
The output functions $\phi_1,\ldots,\phi_\mathsf{D}$ are then uniformly
assigned\footnote{Here we assume for simplicity that
$\mathsf{Q}$ divides $\mathsf{D}$. Note that otherwise we can always add empty functions for the
assumption to hold.}\label{footnote1} to the nodes in $\mathbf{Q}$. Let
$\mathcal{D}_{k}^\mathbf{Q}$ be the set of indices of output functions
assigned to a given node $k\in\mathbf{Q}$. Thus,
$\mathbf{\Gamma}^{\mathbf{Q}}\triangleq
\big\{\mathcal{D}_{k}^{\mathbf{Q}}\big\}_{k\in\mathbf{Q}}$ forms a
partition of $[\mathsf{D}]$,  and each set $\mathcal{D}_{k}^\mathbf{Q}$
is of cardinality $\frac{\mathsf{D}}{\mathsf{Q}}$. \qifa{Denote
the set of all the partitions of $[\mathsf{D}]$ into \qifa{$\sys{Q}$
equal-sized subsets}  by $\mathbf{\Delta}$.}
  \item \textbf{Shuffle Phase:} The nodes in $\mathbf{Q}$ proceed to exchange their computed IVAs. Each node $k\in\mathbf{Q}$ multicasts a signal
 \begin{IEEEeqnarray}{c}
        X_{k}^{{\mathbf{Q}}}=\varphi^{{\mathbf{Q}}}_{k}\left(\mathcal{C}_{ k}, \mathbf{\Gamma}^{\mathbf{Q}}\right)\notag
\end{IEEEeqnarray}
 to all  the other nodes in $\mathbf{Q}$. \mw{For each $k\in\mathbf{Q}$, here
$ \varphi_{k}^{{\mathbf{Q}}}
 \;\colon \;\mathbb{F}_{2}^{|\mathcal{C}_{k}|\mathsf{V}}\times\mathbf{\Delta}\rightarrow
  \mathbb{F}_2^{|X_k^{\mathbf{Q}}|}$
denotes the encoding function of node $k$.}

We assume a perfect multicast channel, i.e.,
each active node $k\in \mathbf{Q}$ receives perfectly all the transmitted signals
\begin{IEEEeqnarray}{c}
X^{{\mathbf{Q}}}\triangleq\big\{ X_{k}^{{\mathbf{Q}}}:
k\in\mathbf{Q}\big\}.\notag
\end{IEEEeqnarray}

  \item \textbf{Reduce Phase:} Using the received signals
    $X^{{\mathbf{Q}}}$ from the shuffle phase and the local IVAs
    $\mathcal{C}_{k}$ computed in the map phase, node $k$ has to be able to compute all the IVAs
 \begin{IEEEeqnarray}{c}
       \left\{(v_{d,1},\ldots,v_{d,\mathsf{N}})\right\}_{d\in\mathcal{D}_{k}^\mathbf{Q}}=\psi^{
    {\mathbf{Q}}}_{k}\left(X^{{\mathbf{Q}}},\mathcal{C}_{ k}, \mathbf{\Gamma}^{\mathbf{Q}}\right),\notag
\end{IEEEeqnarray}
where
$       \psi^{{\mathbf{Q}}}_{k}: \mathbb{F}_{2}^{\sum_{k\in\mathbf{Q}}|X_k^{{\mathbf{Q}}}|}\times \mathbb{F}_{2}^{|\mathcal{C}_{k}|
       \mathsf{V}}\times\mathbf{\Delta}\rightarrow
       \mathbb{F}_2^{\frac{\mathsf{NDV}}{\mathsf{Q}}}$.
Finally, with the restored IVAs, it computes each assigned function via
the reduce function, namely,
 \begin{IEEEeqnarray}{c}
   u_{d}=h_{d}(v_{d,1},\ldots,v_{d,\sys{N}}),\quad \forall~d\in\mathcal{D}_{ k}^\mathbf{Q}.\notag
\end{IEEEeqnarray}

\end{enumerate}

To measure the storage and communication costs, we introduce the following
definitions.

\begin{definition}[Storage Load] \emph{Storage load} $\overline{r}$ is
  defined as the total number of files stored across the $\mathsf{K}$
  nodes normalized by the total number of files $\mathsf{N}$, i.e.,
\begin{IEEEeqnarray}{c}
\overline{r}\triangleq\frac{\sum_{k=1}^K|\mathcal{M}_k|}{\mathsf{N}}.\notag
\end{IEEEeqnarray}
\end{definition}

\begin{definition}[Communication Load] \emph{Communication load} $\overline{L}$ is defined as the average total number of bits sent in the shuffle phase, normalized by the total number of bits of all intermediate values, i.e.,
\begin{IEEEeqnarray}{c}
\overline{L}=\mathbf{E}\left[\frac{\sum_{k\in\mathbf{Q}}|X_k^{\mathbf{Q}}|}{\mathsf{NDV}}\right],\label{communication_load}
\end{IEEEeqnarray}
where the expectation is taken over  the \mw{random} active set $\mathbf{Q}$.
\end{definition}

\begin{definition} [Storage-Communication (SC) Tradeoff] \label{def:SC}A pair of real
  numbers  $(r,L)$ is achievable if for any $\epsilon>0$,  there exist
  positive integers $\sys{N,D,U,V,W}$, a storage design $\{\mathcal{M}_k\}_{k=1}^{\mathsf{K}}$ of storage load less than $r+\epsilon$, a set of uniform assignments of output functions $\left\{\mathbf{\Gamma}^{\mathcal{Q}}\right\}_{\mathcal{Q}\in\mathbf{\Omega}_\sys{K}^{\sys{Q}}}$, and a collection of encoding functions $\big\{ \big\{\varphi^{\mathcal{Q}}_{k}\big\}_{k\in\mathcal{Q}}\big\}_{\mathcal{Q} \in \mathbf{\Omega}_{\sys{K}}^{\sys{Q}}}$ with communication load less than $L+\epsilon$,
	such that all the output functions $\phi_1,\ldots, \phi_\mathsf{D}$ can be  computed successfully.
For a fixed $\mathsf{Q}\in[\mathsf{K}]$, we define the fundamental storage-communication (SC) tradeoff  as
\begin{IEEEeqnarray}{c}
L_{\mathsf{K},\mathsf{Q}}^*(r)\triangleq\inf\left\{L:(r,L)~\text{is achievable}\right\}.\notag
\end{IEEEeqnarray}
\end{definition}

Note that the non-trivial interval for the values of
 $r$ is $[\mathsf{K}-\mathsf{Q}+1,\mathsf{K}]$. Indeed,
if $r>\mathsf{K}$, then each node can store all the files and compute any
function locally. On the other hand, from the
assumption~\eqref{eq:tn}, we have for any feasible scheme
\begin{IEEEeqnarray}{rCl}
\overline{r}
&=&\frac{\sum_{k=1}^\mathsf{K}|\mathcal{M}_k|}{\mathsf{N}}\notag\\
&=&\frac{\sum_{n=1}^{\mathsf{N}}t_n}{\mathsf{N}}\notag\\
&\geq&\mathsf{K}-\mathsf{Q}+1.\notag
\end{IEEEeqnarray}

Therefore, throughout the paper, we only focus on the interval
$r\in[\mathsf{K}-\mathsf{Q}+1,\mathsf{K}]$ for any given
$\sys{Q}\in[\sys{K}]$.
Further,  for a given storage design $\{\mathcal{M}_k\}_{k=1}^{\mathsf{K}}$, by the symmetry assumption of the reduce functions and the fact that each node has all the IVAs of all $\mathsf{D}$ output functions in the files it has stored,   the optimal communication load is independent  of  the reduce function assignment.
This is similar to the case without straggling nodes (see \cite[Remark 3]{Li2018Tradeoff}).


\begin{definition}[File Complexity] The smallest number of files $\sys{N}$ required to implement a given scheme is called the file complexity
of this scheme.
\end{definition}

\mw{In above problem definition,  the various nodes store entire  files during the map phase  and during the shuffle phase they reconstruct all the  IVAs corresponding to their output functions. This system definition does not allow to reduce  storage or communication loads by exploiting special structures of the map or reduce functions as proposed for example in \cite{Li2016unified,Zhang2018Improved}. As a consequence, all the coded computing schemes presented in this paper universally apply to arbitrary map and reduce functions and the SC tradeoff in Definition~\ref{def:SC} applies only to such universal schemes. In fact, as we will explain, for linear reduce functions the SC tradeoff derived in \cite{Zhang2018Improved} improves over the one  in Definition~\ref{def:SC}, since it was derived for a system where nodes  do not have to store  individual files and reconstruct all the required IVAs, but linear combinations of them suffice.
}
\section{Placement Delivery Arrays for Straggling
Systems}\label{sec:PDA}


\subsection{Definitions}
Placement delivery arrays (PDA) introduced in \cite{Yan2017PDA} are the main tool of this paper. To adapt to our setup, we use the \qifa{following} definition \mw{from} \cite{Yan2018SCC}.
\begin{definition}[Placement Delivery Array (PDA)]\label{def:PDA} For positive
  integers $\sys{K,F,T}$ and a nonnegative integer $\sys{S}$,  an $\mathsf{F}\times \mathsf{K}$ array
  $\bm{A}=[a_{j,k}]$, $j\in [\mathsf{F}], k\in[\mathsf{K}]$, composed of
  $\mathsf{T}$ special symbols $``*"$  and some ordinary symbols $1,\ldots, \mathsf{S}$, each occurring at least once,  is called   a $(\mathsf{K,F,T,S})$ PDA, if for any two distinct entries $a_{j_1,k_1}$ and $a_{j_2,k_2}$,   we have $a_{j_1,k_1}=a_{j_2,k_2}=s$ \mw{with $s$} an ordinary symbol  only if
  \begin{enumerate}
     \item [a)] $j_1\ne j_2$, $k_1\ne k_2$, i.e., they lie in distinct rows and distinct columns; and
     \item [b)] $a_{j_1,k_2}=a_{j_2,k_1}=*$,
     \end{enumerate}
     \mw{ i.e., the corresponding $2\times 2$  subarray formed by rows $j_1,j_2$ and columns $k_1,k_2$ must be of the following form}
  \begin{IEEEeqnarray}{c}
    \left[\begin{array}{cc}
      s & *\\
      * & s
    \end{array}\right]~\textrm{or}~
    \left[\begin{array}{cc}
      * & s\\
      s & *
    \end{array}\right].\notag
  \end{IEEEeqnarray}

   A PDA with all ``$*$'' entries is called trivial. Notice that in this case $\sys{S} = 0$ and $\sys{KF} = \sys{T}$. A PDA  is \qifa{called a $g$-regular PDA} if each ordinary symbol occurs exactly $g$ times.
\end{definition}
\begin{example}\label{exam:PDA} The following array is a $3$-regular $(4,6,12,4)$ PDA.
\begin{IEEEeqnarray}{c}
\bm{A}=\left[\begin{array}{cccc}
        * & * & 1 & 2 \\
        * & 1 & * & 3 \\
        * & 2 & 3& * \\
        1 & * & * & 4 \\
        2 & * & 4 & * \\
        3 & 4 & * & *
      \end{array}
\right].\notag
\end{IEEEeqnarray}
\end{example}
For our purpose, we introduce the following definitions similarly to the ones in \cite{Yan2018SCC}.
\begin{definition}[PDA for Distributed Computing (Comp-PDA)]
A Comp-PDA is a PDA with at least one $``*"$-\mw{symbol} in each  row.
\end{definition}

\begin{definition}[Minimum Storage Number]
Given a  Comp-PDA $\bm A$, its \emph{minimum storage number} $\tau$ is
defined as the minimum number of $``*"$-symbols in any of the rows of $\bm A$.
\end{definition}

\begin{definition}[Symbol Frequencies]
 For a given nontrivial $(\sys{K}, \sys{F}, \sys{T}, \sys{S})$ Comp-PDA, let $\sys{S}_t$ denote the number of ordinary symbols that occur
exactly $t$ times, for $t \in[\sys{K}]$. The symbol frequencies $\theta_1, \theta_2, \ldots,\theta_\sys{K}$ of the Comp-PDA are then defined as
\begin{IEEEeqnarray}{c}
\theta_t\triangleq\frac{\sys{S}_tt}{\sys{K}\sys{F}-\sys{T}},\quad t\in[\sys{K}].\notag
\end{IEEEeqnarray}
They indicate the fractions of ordinary entries of the Comp-PDA that occur exactly $1,2,\ldots,\sys{K}$ times, respectively. For completeness, we also define $\theta_t\triangleq 0$ for $t>K$.
\end{definition}

\subsection{Constructing a Coded Computing Scheme \sheng{from} a Comp-PDA: A Toy Example}\label{CCS:example}

\mw{In this subsection we illustrate the connection between
Comp-PDAs and  coded computing schemes with stragglers at hand of  \mw{a} toy example. Section  \ref{sec:CC-PDA} ahead describes a general procedure to obtain a coded computing scheme with stragglers from any Comp-PDA, and it  presents a performance
analysis for the obtained scheme.}

Consider the  $(4,6,12,4)$ Comp-PDA $\bm A$ in Example \ref{exam:PDA},
and assume a $(\sys{K},\sys{Q})=(4,3)$ straggling system with
$\sys{N}=6$ files and $\sys{D}=3$ output functions. The scheme is
illustrated in Fig.~\ref{fig:example} for the case that node $3$ is
straggling. \mw{In this Fig.~\ref{fig:example}, the  line ``files'' in each of the four boxes indicates the  files stored at the nodes. The remaining lines in the boxes illustrate \mw{the computed IVAs, where red circles, green triangles, and blue squares depict IVAs pertaining to output functions $\phi_1$, $\phi_2$, and $\phi_{3}$, respectively. More specifically, a red circle with the number $i$ in the middle stands for IVA $v_{i,1}$, and so on.} The lines below the boxes of the active nodes 1, 2, and 4 indicate the IVAs that  the nodes have to learn to be able to compute their output functions. In this example it is assumed that  node 1 computes $\phi_1$, node~2 computes $\phi_2$, and node $4$ computes $\phi_3$.   
The signals on the left/right side of the boxes indicate the signals sent by the nodes. Here,  splitting of IVAs indicates that the IVA is decomposed into a substring consisting of the first half of the bits and a substring consisting of the second half of the bits, and  the plus symbol stands for a  bit-wise XOR-operation on the substrings.}

\mw{We now explain the distributed coding scheme associated with the PDA
${\bm A}$ more formally.  We start by associating   column $k$ of ${\bm
A}$  with node $k$  in the system, ($k\in[4]$), and row $j$ of ${\bm A}$
with file $w_j$ in the system, ($j\in[6]$). In the map phase,  node $k$
stores  file $w_j$  if the row-$k$ and column-$j$ entry of ${\bm A}$ is
a  $``*"$-symbol. For example, node $1$, which is associated with the first column of the Comp-PDA, stores  files $w_1, w_2$ and $w_3$.
Each node then computes all the IVAs corresponding to the files it has stored.}

\begin{figure}
  \centering
  \includegraphics[width=0.48\textwidth]{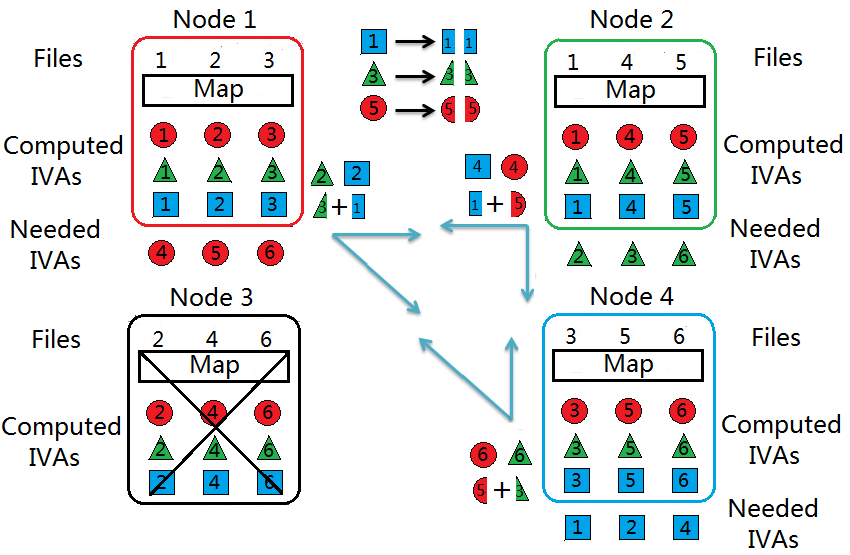}\\
  \caption{An example of CCS scheme for a system with $\mathsf{K}=4$, $\mathsf{N}=6$ and  $\mathsf{Q}=3$, where the third node is a straggling node.}\label{fig:example}
\end{figure}

\mw{In our example, we assume that  node $3$ is the only straggler. Nodes 1, 2, and 4 thus form the active set and as such continue with the shuffle and reduce procedures. Accordingly, we  extract from the PDA ${\bm A}$ the subarray ${\bm A}^{\{1,2,4\}}$  consisting of columns $1,2$ and $4$ (the columns corresponding to the active set):
	\begin{IEEEeqnarray}{c}
{\bm{A}}^{\{1,2,4\}}=\left[\begin{array}{ccc}
        * & * & 2 \\
        * & 1 &   3 \\
        * & 2 &  * \\
        1 & * &   4 \\
        2 & * &   * \\
        3 & 4 &   *
      \end{array}
\right].\notag
\end{IEEEeqnarray}}

\mw{Notice that  ${\bm A}^{\{1,2,4\}}$ is also a Comp-PDA \maw{(in particular it has at least one ``$*$" symbol in each row)} and the node
corresponding to a given column has stored all the files indicated by
the $``*"$-symbols in this column. \maw{The same statement applies also to the subarrays associated with any other possible active set of size $3$. After the shuffling phase, we are thus in the same situation as  described in
 \cite{Yan2018SCC,Tang2018ITW} when a coded computing scheme without stragglers is to be constructed from a Comp-PDA, and as a consequence,  the same shuffle and reduce procedures  can be applied.} We described these procedures here in detail for completeness.

 The shuffle phase is as follows. For each $s\in\{1,2,3,4\}$ occuring $g$ times ($g=2$ or $3$), pick out the $g\times g$ array containing $s$. For example, symbol  $s=2$ is associated with the following 3-by-3 subarray:}
\begin{IEEEeqnarray}{c}
\begin{array}{c|ccc}\bottomrule
&1&2&4\\\hline
       w_1& * & * & 2 \\
       w_3& * & 2 & * \\
       w_5& 2 & * & *\\\toprule
      \end{array}\label{sub:array}
\end{IEEEeqnarray}
Each \mw{occurence of the symbol $``2"$ in this subarray stands for an
IVA desired by the node in the corresponding column and computed at the
other nodes in this subarray. The row of the symbol indicates to which
file this IVA pertains, and the $``*"$ symbols in this row indicate
that the IVA can indeed be computed by all nodes \maw{in the active set} except for the one
corresponding to the column of the $``2"$ symbol. In the above example,
the three  $``2"$ symbols from top to down represent the IVAs
$v_{3,1}$, $v_{2,3}$, and $v_{1,5}$, respectively. These IVAs are
shuffled in a coded manner. To this end, they are first split into
$g-1=2$ equally-large sub-IVAs, and each of these sub-IVAs is labeled by
one of the nodes where the IVA has been computed (i.e,. by the columns
with $``*"$ symbols). In our example, we split $v_{3,1}=(v_{3,1}^1,v_{3,1}^2)$, $v_{2,3}=(v_{2,3}^1,v_{2,3}^4)$ and $v_{1,5}=(v_{1,5}^2,v_{1,5}^4)$. The signal sent by a given node $i$ is then simply the componentwise XOR of the sub-IVAs with superscript $i$. So,
  nodes $1$, $2$, $4$ send $v_{2,3}^1\oplus v_{3,1}^1$, $v_{3,1}^2\oplus v_{1,5}^2$ and $v_{1,5}^4\oplus v_{2,3}^4$, respectively. The same procedure is applied for  all other ordinary symbols $1$, $3$, and $4$ in subarray ${\bm A}^{\{1,2,4\}}$. The following table \qifa{lists} all the signals sent at the 4 nodes, \qifa{where the first line lists their associated \maw{ordinary} symbols}:}
  \qifa{
  \begin{IEEEeqnarray}{c}
\begin{array}{c|c|c|c|c}\bottomrule
\textnormal{Symbol}&1&2&3&4\\\hline
\textnormal{Node } 1&v_{2,2}& v_{2,3}^1\oplus v_{3,1}^1 &v_{3,2}& \\\hline
   \textnormal{Node }  2   &v_{1,4}&  v_{3,1}^2\oplus v_{1,5}^2&  &v_{3,4}\\\hline
       \textnormal{Node } 3 &\multicolumn{4}{c}{\textnormal{(straggling)}} \\\hline
       \textnormal{Node } 4 &&v_{1,5}^4\oplus v_{2,3}^4 &v_{1,6}  & v_{2,6}\\\toprule
      \end{array}
\end{IEEEeqnarray}
}
\mw{We now explain how the nodes extract their missing IVAs from the
shuffled signals.} Since node $1$ has computed $v_{1,1},v_{1,2}$ and
$v_{1,3}$ in the map phase,  it still needs to decode
$v_{1,4},v_{1,5},v_{1,6}$. \qifa{It \maw{directly  obtains} the IVAs $v_{1,4}$ and
$v_{1,6}$ from the uncoded signals sent by nodes $2$ and $4$
respectively.} \maw{Moreover, it reconstructs the two sub-IVAs $v_{1,5}^{2}$ and $v_{1,5}^{4}$, by
XORing the  signals $v_{3,1}^{2}\oplus v_{1,5}^{2}$ and
$v_{1,5}^{4}\oplus v_{2,3}^{4}$ shuffled by nodes $2$ and $4$ with its locally stored sub-IVAs $v_{3,1}^{2}$ and $v_{2,3}^{4}$.}
\maw{Nodes $2$ and $4$ reconstruct their missing IVAs in a similar way.}
%

\qifa{The total number of stored files at the nodes is $3\times 4=12$, thus the storage load is $r=\frac{3\times K}{N}=2$.} The total length of the transmitted signals is $7.5\mathsf{V}$, which remains unchanged also when any of the other nodes straggles. 
The communication load is thus $L=\frac{7.5\mathsf{V}}{6\times 3\times \mathsf{V}}=\frac{5}{12}$.

\section{Main Results}\label{sec:mainresults}
In this section, we present our main results. Details and proofs are deferred to Sections \ref{sec:CC-PDA}--\ref{sec:proofofthmF}.
\subsection{Coded Computing Schemes for Straggling Systems from Comp-PDAs }
In Section \ref{sec:CC-PDA}, we propose  a coded computing scheme for a $(\sys{K},\sys{Q})$ straggling system based on any Comp-PDA with $\sys{K}$ columns and minimum storage number $\tau\geq\sys{K}-\sys{Q}+1$. Theorem \ref{thm:PDA} is proved by analyzing the coded computing scheme, which is deferred to Section \ref{subsec:performancePDA}.

\begin{theorem}\label{thm:PDA}
  From any given $\mathsf{(K,F,T,S)}$ Comp-PDA $\bm A$ with symbol \qifa{frequencies} $\{\theta_t\}_{t=1}^{\sys{K}}$ and minimum storage number
 $\tau\in[\sys{K-Q+1}:\sys{K}]$, one can construct a coded computing scheme for a $(\sys{K},\sys{Q})$ straggling system achieving
 the SC pair
\begin{IEEEeqnarray}{rCl}
r_{\bm A}&=&\mathsf{\frac{T}{F}},\notag\\
L_{\bm A}&=&\left(1-\frac{\sys{T}}{\mathsf{FK}}\right)\cdot\frac{1}{C_{\mathsf{K}-1}^{\mathsf{Q}-1}}\cdot\sum_{t=1}^\mathsf{K} \theta_t\left(C_{\mathsf{K}-t}^{\mathsf{Q}-1}+\sum_{l=\max\{1,t-\mathsf{K+Q}-1\}}^{\min\{t,\mathsf{Q}\}-1}\frac{1}{l}\cdot C_{t-1}^l\cdot C_{\mathsf{K}-t}^{\mathsf{Q}-l-1}\right),\label{LA:PDA}
\end{IEEEeqnarray}
with file complexity $\sys{F}$.
\end{theorem}

 Theorem \ref{thm:PDA} characterizes the performance of the coded computing scheme obtained from a Comp-PDA as described in Section \ref{sec:CC-PDA} in terms of the Comp-PDA parameters. In the following, we will simply say that a Comp-PDA achieves this performance.

 Notice that the file complexity of any Comp-PDA based scheme coincides with the number of rows $\sys{F}$ of the Comp-PDA. We shall therefore call the parameter $\sys{F}$ of a Comp-PDA its file complexity.

 As we show in the following, Theorem \ref{thm:PDA} can be simplified for regular Comp-PDAs.

\begin{corollary}\label{corollary:regular} From any given $g$-regular $\sys{(K,F,T,S)}$ Comp-PDA $\bm A$, with
   $g\in[\sys{K}]$ and  minimum storage number  $\tau\in[\sys{K-Q+1}:\sys{K}]$, one can construct a coded computing scheme for a $(\sys{K,Q})$
  straggling system achieving thre SC pair
\begin{IEEEeqnarray}{rCl}
r_{\bm A}&=&\sys{\frac{T}{F}},\notag\\
L_{\bm A}&=&\left(1-\sys{\frac{T}{KF}}\right)\cdot\left(\frac{C_{\sys{K}-g}^{\sys{Q}-1}}{C_{\sys{K}-1}^{\sys{Q}-1}}+\sum_{l=\max\{1,g-\sys{K+Q}-1\}}^{\min\{g,\sys{Q}\}-1}\frac{1}{l}\cdot\frac{C_{g-1}^l\cdot C_{\sys{K}-g}^{\sys{Q}-l-1}}{C_{\sys{K}-1}^{\sys{Q}-1}}\right),\notag
\end{IEEEeqnarray}
with file complexity $\sys{F}$.
\end{corollary}
\begin{IEEEproof} From Theorem \ref{thm:PDA}, we only need to evaluate $L_{\bm A}$ when $\bm A$ is a $g$-$\sys{(K,F,T,S)}$ Comp-PDA. In this case,  all the $\sys{S}$ symbols occur $g$ times,  i.e.,
\begin{IEEEeqnarray}{c}
\theta_g=1,\quad\textnormal{and}\quad\theta_{t}=0,\quad\forall~t\in[\sys{K}]\backslash\{g\}.\notag
\end{IEEEeqnarray}
Then the conclusion directly follows from Theorem \ref{thm:PDA}.
\end{IEEEproof}
Corollary \ref{corollary:regular} is of particular interest since there
are several explicit regular PDA constructions for coded caching in the
literature, such as \cite{Yan2017PDA,HyperGraph,bipartite}, which are
also Comp-PDAs. 
 In particular, the following PDAs obtained from the coded
caching scheme proposed by Maddah-Ali and Niesen
\cite{Maddah2014Fundamental} are important.

\begin{definition}[Maddah-Ali Niesen PDA (MAN-PDA)]\label{def:MNPDA}
	Fix  any integer  $i\in[\sys{K}]$, and
	let \new{ $\{\mathcal{T}_j\}_{j=1}^{ C_\sys{K}^i}$} denote  all subsets of $[\sys{K}]$ of size~$i$. Also, choose an arbitrary bijective function $\kappa$ from the collection of all subsets of $[\sys{K}]$ with cardinality $i+1$ 
to the set \new{$\big[C_\sys{K}^{i+1}\big]$}.
Then, define the array $\bm
	P_i=[p_{j,k}]$ as
	\begin{IEEEeqnarray}{c}
		p_{j,k}\triangleq \left\{\begin{array}{ll}
			*, &\textnormal{if}~k\in\mathcal{T}_j \\
			\kappa(\{k\} \cup \mathcal{T}_{j}), &\textnormal{if}~k\notin\mathcal{T}_j
		\end{array}
		\right..\notag
	\end{IEEEeqnarray}
\end{definition}

We observe that for any $i \in [\sys{K}-1]$, the array $\bm P_i$ is an $(i+1)$-regular
\new{$\big(\sys{K},C_\sys{K}^i, \sys{K}{C_{\sys{K}-1}^{i-1}},{C_\sys{K}^{i+1}}\big)$} Comp-PDA (see \cite{Yan2017PDA} for
details). For $i=\sys{K}$, the Comp-PDA $\bm P_i$ consists only of
``$*$''-entries and is thus a trivial PDA. By Corollary~\ref{corollary:regular}, we directly obtain the following result.

\begin{corollary}\label{corollary:MANPDA} For a $(\sys{K},\sys{Q})$ straggling system, and each \mw{$r$ in the discrete set} $[\sys{K}-\sys{Q}+1:\sys{K}]$, the MAN-PDA $\bm P_r$ achieves the SC pair $(r,L_{\bm P_r})$, where
\begin{IEEEeqnarray}{c}
L_{\bm P_r}\triangleq\left(1-\frac{r}{\mathsf{K}}\right)\cdot\sum_{l=r+\mathsf{Q}-\mathsf{K}}^{\min\{r,\mathsf{Q}-1\}}\frac{1}{l}\cdot\frac{C_{r}^l\cdot C_{\mathsf{K}-r-1}^{\mathsf{Q}-l-1}}{C_{\mathsf{K}-1}^{\mathsf{Q}-1}}.\notag
\end{IEEEeqnarray}
\end{corollary}

The coded computing scheme associated to $\bm P_r$ is
equivalent to our proposed \emph{coded computing for straggling systems
(CCS)} in \cite{YanISIT2019}. Here, we present it as a
special case of the more general Comp-PDA framework. As we shall
see, the Comp-PDA framework allows us to design new coded computing
schemes with lower file complexity.

\subsection{The Fundamental Storage-Communication Tradeoff}

We are ready to present  our  result on the fundamental SC tradeoff, which is  proved in Section \ref{sec:SCtradeoff}.

\begin{theorem}\label{thm:limits}
  For a $(\sys{K,Q})$ straggling system, with a given integer storage
  load
  \mw{$r$ in the discrete set} $[\mathsf{K}-\mathsf{Q}+1:\mathsf{K}]$, the fundamental
SC tradeoff is
\begin{IEEEeqnarray}{c}
L_{\mathsf{K},\mathsf{Q}}^*(r)=\left(1-\frac{r}{\mathsf{K}}\right)\cdot\sum_{l=r+\mathsf{Q}-\mathsf{K}}^{\min\{r,\mathsf{Q}-1\}}\frac{1}{l}\cdot\frac{C_{r}^l\cdot C_{\mathsf{K}-r-1}^{\mathsf{Q}-l-1}}{C_{\mathsf{K}-1}^{\mathsf{Q}-1}}, \qquad r\in[\mathsf{K}-\mathsf{Q}+1:\mathsf{K}], \label{optimal:LQ}
\end{IEEEeqnarray}
which is achievable with a scheme of file complexity $C_{\sys{K}}^r$. \mw{For a general $r$ in the interval $[\mathsf{K}-\mathsf{Q}+1,\mathsf{K}]$,} the fundamental SC tradeoff $L_{\mathsf{K},\mathsf{Q}}^*(r)$ is given by the lower convex envelope formed by the above points in \eqref{optimal:LQ}.
\end{theorem}

Fig.~\ref{fig:tradeoff} shows the fundamental SC tradeoff curves for $\sys{K}=10$ and different
 values of $\sys{Q}$. When $\mathsf{Q}=1$, the curve reduces to a single point $(\mathsf{K},0)$,
while when $\mathsf{Q}=\mathsf{K}$, the curve corresponds to the fundamental tradeoff
without straggling nodes~(cf.~\cite[Fig. 1]{Li2018Tradeoff}).  In this latter case without stragglers, the fundamental SC tradeoff curve is achieved by the CDC scheme in
\cite{Li2018Tradeoff}.
\mw{For a general value of $\mathsf{Q}$ and   integer storage $r\in[\sys{K}-\sys{Q}+1:\sys{K}]$, the fundamental SC tradeoff pair $(r,L_{\sys{K},\sys{Q}}^*(r))$ is  achieved by the MAN-PDA $\mathbf{P}_r$, see
 Corollary~\ref{corollary:MANPDA}.
This implies that for a fixed integer storage load $r\in[1:\sys{K}]$,
the SC pairs $\{(r,L_{\sys{K},\sys{Q}}^*(r))\}_{\sys{Q}=\sys{K}-r+1}^K$
are all achieved by the same PDA  $\mathbf{P}_r$, irrespective of the
size of the active set $\sys{Q}$. As we show in
Section~\ref{subsec:connection}, the map procedures  of the  coded
computing scheme corresponding to a given Comp-PDA at a given node $k$
only depends  on the $``*"$-symbols in the $k$-th column of the PDA. Therefore, all the points on the fundamental SC tradeoff curve with same  integer storage load $r$ can be attained with the same  map procedures described by the MAN-PDA $\mathbf{P}_r$. (See also Remark~\ref{remark:PDAmap} in Section~\ref{subsec:connection}.)}


\maw{As a consequence, the fundamental SC-tradeoff points that have integer storage load  $r\in[1,\sys{K}]$ remain achievable (and optimal) also  in a related setup where the size of the active set $\sys{Q}$ is unknown  during the map procedure. By simple time and memory-sharing arguments, this conclusion extends to all points on the fundamental SC tradeoff curve with arbitrary real-valued storage loads $r\in[1,\sys{K}]$. Related is also the scenario where instead of fixing the size of the active set $\sys{Q}$, the system  imposes a hard time-limit for the map phase and proceeds to the shuffle and reduce phases with the (random) number of nodes that have terminated within due time. For given storage load $r$, the MAN-PDA based coded computing scheme promises that  when $\sys Q \geq \lceil \sys{K}+1-r \rceil$ nodes have terminated during the map phase, all IVAs are computed at least once and thus the system can proceed to data shuffling, and achieves the minimum required communication load  $L_{\sys{K},\sys{Q}}^*(r)$. When only $\sys Q < \lceil \sys{K}+1-r \rceil$ nodes have terminated, some IVAs are not computed, and hence the system can not proceed. }


 \begin{figure}[htbp]
   \centering
   \includegraphics[width=0.49\textwidth]{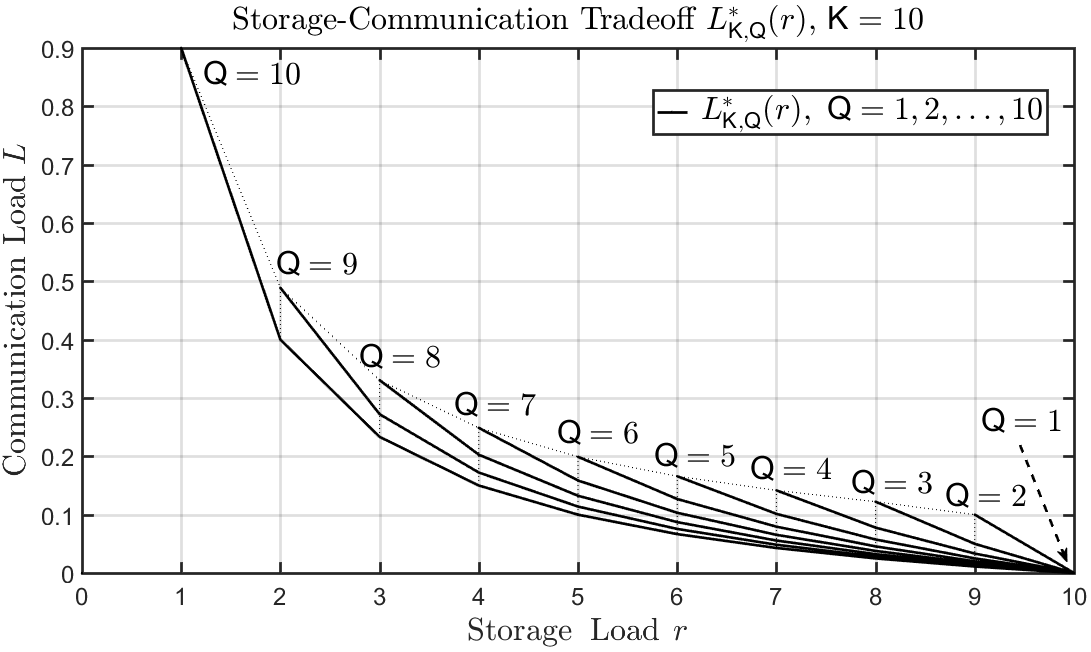}
   \caption{Storage-Communication Tradeoff  $L_{\mathsf{K},\mathsf{Q}}^*(r)$ for $\mathsf{Q}\in[\mathsf{K}]$
   when  $\mathsf{K}=10$. }
   \label{fig:tradeoff}
   \end{figure}

It is \maw{further} worth pointing out that all our PDA based coded computing schemes \mw{are universal and achieve the same performance for any choice of  map and reduce functions.
No structure is assumed on these functions. 
Similarly, our information-theoretic converse applies only to such universal coded computing schemes.} If the map  or reduce
functions have certain properties, for example, linearity,
  it is possible to achieve
better SC tradeoffs \mw{by storing combinations of files instead of each file separately \cite{Li2016unified,Zhang2018Improved}.} Fig.
\ref{fig:comparison} compares Theorem \ref{thm:limits} to the results in
\cite{Li2016unified, Zhang2018Improved}. It can be observed that the MAN-PDA based
scheme outperforms the scheme in \cite{Li2016unified} but is inferior to
the improved version in \cite{Zhang2018Improved}. \mw{As already mentioned, the scheme in \cite{Zhang2018Improved} however works only for linear map functions, and not for arbitrary functions as our schemes. Another advantage of our schemes is that they work over the binary field, and are thus easier to implement than the MDS-based schemes in
 \cite{Li2016unified, Zhang2018Improved}. 
which require a large enough field size.}
\begin{figure}[htbp]
   \centering
   \includegraphics[width=0.49\textwidth]{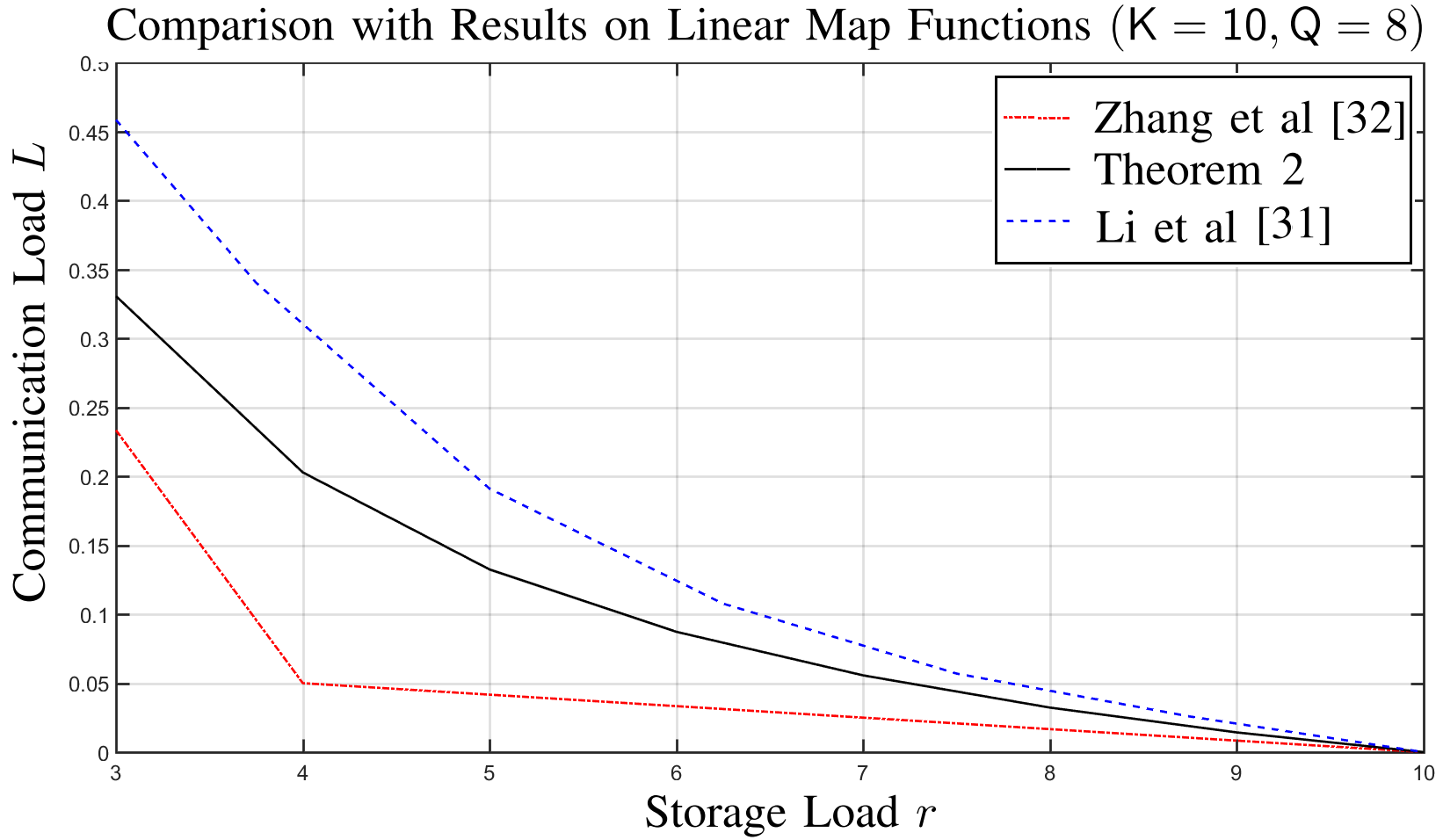}
   \caption{Comparison with known results when applied to linear map  functions, $\mathsf{K}=10, \mathsf{Q}=8$.  }
   \label{fig:comparison}
   \end{figure}

%
%
%
%
\subsection{Optimality and Reduction of File Complexity}
From Theorem \ref{thm:PDA} and Corollary \ref{corollary:MANPDA}, the
coded computing scheme based on the MAN-PDA $\bm P_r$, for
$r\in[\sys{K-Q+1}:\sys{K}]$, has file complexity $\new{\sys{F}=C_\sys{K}^r}$ and achieves the fundamental SC tradeoff.  The following theorem indicates that, this is the smallest file complexity to achieve the same tradeoff point in most cases. The proof is deferred to Section~\ref{sec:proofofthmF}.

\begin{theorem}\label{thm:optimalF}
  For a $(\sys{K,Q})$ straggling system, if
 a Comp-PDA based scheme achieves the fundamental SC tradeoff
 $\big(r,L_{\sys{K,Q}}^*(r)\big)$ for some
 $r\in[\sys{K-Q}+1:\sys{K}]$, if $\sys{Q}\notin\{2,\sys{K}\}$ or $r\neq\sys{K}-\sys{Q}+1$,  then the file complexity
 $\sys{F}\geq  C_\sys{K}^r$.
\end{theorem}
\begin{remark} It is easy to verify that, in the case $\sys{Q}\in\{2,\sys{K}\}$ and $r=\sys{K}-\sys{Q}+1$, the fundamental SC tradeoff can be achieved with $\sys{F}=1$ with the  Comp-PDAs $[*,*,\ldots,*,1]$ and $[*,1,2,\ldots,\sys{K}-1]$, respectively.
\end{remark}

We next present Comp-PDAs with lower file complexity $\sys{F}$ that achieve SC tradeoffs close to the optimal ones.
 We consider two existing PDA constructions originally proposed for coded
caching in \cite[Theorems 4 and 5]{Yan2017PDA}.
 Let $q\in[2:\sys{K}-1]$ be such that $q\,|\,\sys{K}$, and $m=\frac{\sys{K}}{q}$. There exists
\begin{enumerate}
  \item [P$1$)]an $m$-regular $\big(mq,q^{m-1},mq^{m-1},(q-1)q^{m-1}\big)$ Comp-PDA with minimum storage number $m$;
  \item [P$2$)]an $m(q-1)$-regular $\big(mq,(q-1)q^{m-1},m(q-1)^2q^{m-1},q^{m-1}\big)$ Comp-PDA with minimum storage number $m(q-1)$.
\end{enumerate}

\begin{corollary}\label{cor}
  For any integer $r\in[\sys{K-Q}+1:\sys{K}-1]$, such that either $r\,|\,\sys{K}$ or $\qifa{(\sys{K}-r)}\,|\,\sys{K}$, the communication load
\begin{IEEEeqnarray}{c}
L_{\sys{K},\sys{Q}}(r)=\left(1-\frac{r}{\sys{K}}\right)\cdot\left(\frac{C_{\sys{K}-r}^{\sys{Q}-1}}{C_{\sys{K}-1}^{\sys{Q}-1}}+\sum_{l=\max\{1,r+\sys{Q-K}-1\}}^{\min\{r,\sys{Q}\}-1}\frac{1}{l}\cdot\frac{C_{r-1}^lC_{\sys{K}-r}^{\sys{Q}-l-1}}{C_{\sys{K}-1}^{\sys{Q}-1}}\right),\label{LQ}
\end{IEEEeqnarray}
can be achieved with file complexity  $\sys{F} = \frac{r}{\sys{K}}\cdot\left(\frac{\sys{K}}{\min\{r,\sys{K}-r\}}\right)^{\min\{r,\sys{K}-r\}}$.
\end{corollary}
\begin{IEEEproof}  If $r\,|\,\sys{K}$, then specialize the Comp-PDA in P1) to parameter $q=\frac{\sys{K}}{r}$. This results in a $r$-regular $\left(\sys{K},q^{r-1},\sys{K}q^{r-2},(q-1)q^{r-1}\right)$ Comp-PDA with minimum storage number $r$, and the proof is then immediate from Corollary \ref{corollary:regular}. If $\sys{K}-r|\sys{K}$, then  specialize the Comp-PDA in P2) to parameter $q=\frac{\sys{K}}{\sys{K}-r}$. This results in a $r$-regular $\big(\sys{K},(q-1)q^{\sys{K}-r-1},$ $\sys{K}(q-1)^2q^{\sys{K}-r-2},$ $q^{\sys{K}-r-1}\big)$ Comp-PDA, and the proof again follows from Corollary \ref{corollary:regular}.
\end{IEEEproof}

In the following proposition, we quantify how close the above SC
tradeoff point is to the optimal, and \mw{by} how much we can reduce the
file complexity.
\begin{proposition}\label{proposition}
  Consider a $(\sys{K},\sys{Q})$ straggling system and an integer $r\in[\sys{K}-\sys{Q}+1:\sys{K}]$ such that
  $\frac{r}{\sys{K}} = c \in\big\{
  \frac{1}{q}, \frac{q-1}{q}\big\}$ for some integer
  $q\in[2:\sys{K}-1]$. There exist
  $\alpha\in[0, 2]$ and $\beta\in[0,\sqrt{2\pi}e^2]$, such that the SC tradeoff $L_{\sys{K},\sys{Q}}(r)$ and the file complexity $\sys{F}$ achieved by  constructions P1) or P2)  above satisfy
  \begin{IEEEeqnarray}{c}
    \frac{L_{\sys{K},\sys{Q}}(r)}{L_{\sys{K},\sys{Q}}^*(r)} = 1+\frac{\alpha}{r},\notag
  \end{IEEEeqnarray}%
   and
  \begin{IEEEeqnarray}{c}
    \frac{\sys{F}}{\sys{F}^*} = \beta A_q \sys{K}^{\frac{1}{2}} B_q^{-\sys{K}},\notag
  \end{IEEEeqnarray}%
  where $A_q \triangleq \frac{\sqrt{q-1}}{cq}$ and $B_q \triangleq
    \left(\frac{q}{q-1}\right)^{\frac{q-1}{q}}$.
\end{proposition}

\qifa{The proof is given in Appendix \ref{app:B}.} From the above proposition, \maw{for a fixed integer $q$, whenever $\sys{\frac{r}{K}}\in\{\frac{1}{q},\frac{q-1}{q}\}$ and  $\sys{K},r$} scale  proportionally to infinity, the
communication load is close to optimal, while the file complexity can be reduced by a factor that increases exponentially \maw{in}
$\sys{K}$.

\begin{remark} In this work, we only consider two particular PDAs. There has been extensive research in coded caching
  schemes with low subpacketization level using various approaches. Most
  of them have an equivalent PDA representation. For examples, PDAs can be constructed from hyper-graphs \cite{HyperGraph}, bipartite graphs \cite{bipartite}, linear block codes \cite{linearcodes}, Ruzsa-Szemer$\acute{\mbox{e}}$di graphs \cite{Unified}.
  \qifa{The result in Theorem \ref{thm:PDA}}
  makes it possible to apply all these results straightforwardly to straggling systems.
\end{remark}

\section{Coded Computing Schemes for Straggling Systems from Comp-PDAs (Proof of Theorem \ref{thm:PDA})}\label{sec:CC-PDA}

In this section, we prove Theorem \ref{thm:PDA} by describing how
to construct a coded computing scheme from a given Comp-PDA and analyzing its performance.

\subsection{Constructing a Coded Computing Scheme for a  Straggling System from a Comp-PDA}
\label{subsec:connection}

In \cite{Yan2018SCC}, we described how \maw{to obtain a coded computing scheme without stragglers from any given Comp-PDA}. A similar procedure is possible in the
presence of stragglers if the minimum storage number $\tau\geq
\sys{K}-\sys{Q}+1$. In fact, assume a given Comp-PDA $\bm A$. The
storage design in the map phase corresponding to $\bm A$ is the same as
without straggling nodes. As part of the map phase, each node computes
all the IVAs that it can compute from its stored files. For the reduce
phase of the straggling system, we restrict to the subarray $\bm
A^{\mathbf{Q}}$ of $\bm A$ formed by the columns of $\bm A$ with indices
in the active set ${\mathbf{Q}}$. Notice that  $\bm A^{\mathbf{Q}}$ is
again a Comp-PDA, because the minimum storage number of $\bm A$ is at
least $\sys{K}-\sys{Q}+1$ \mw{and after eliminating $\sys{K}-\sys{Q}$
columns from $\bm A$ each row still contains at least one $``*''$ symbol}. Shuffle and reduce phases are  performed as in a non-straggling setup, see \cite{Yan2018SCC}, but where the Comp-PDA $\bm A$ is replaced by the new Comp-PDA  $\bm A^{\mathbf{Q}}$.  For completeness, we explain the map, shuffle, and reduce phases in  detail.

 Fix a $(\mathsf{K,F,T,S})$ Comp-PDA $\bm A=[ a_{ i, j}]$  with minimum storage number $\tau\geq\mathsf{K-Q}+1$. Partition the $\mathsf{N}$ files into $\mathsf{F}$  batches $\mathcal{W}_1,\mathcal{W}_2,\ldots,\mathcal{W}_{\sys{F}}$, each
containing
\begin{IEEEeqnarray}{c}
\eta\triangleq \mathsf{\frac{N}{F}}\notag
\end{IEEEeqnarray}
files and so that $\mathcal{W}_1,\mathcal{W}_2,\ldots,\mathcal{W}_{\sys{F}}$ form a partition for $\mathcal{W}$. It is implicitly assumed here that $\eta$ is an integer number.

\subsubsection{Map Phase}
Each node $k$ stores the files in
\begin{IEEEeqnarray}{c}
\mathcal{M}_{
k}=\mathop\bigcup\limits_{i\in[\mathsf{F}]\,:\,a_{i,k}=*}\mathcal{W}_{i},\label{PDA:Mk}
\end{IEEEeqnarray}
and computes the IVAs in \eqref{set:computedIVAs}. The map phase terminates whenever any $\sys{Q}$ nodes accomplish their computations.  \mw{Throughout this section, let $\mathbf{Q}=\mathcal{Q}$ be the realization of the active  set}. Then,  $\bm A^{\mathcal{Q}}$ denotes the subarray of $\bm A$ composed of the columns in $\mathcal{Q}$. Also, let $g_s^{\mathcal{Q}}$ denote the number of occurrences of the symbol $s$ in  $\bm A^{\mathcal{Q}}$,  i.e.,
\begin{IEEEeqnarray}{c}
g_s^{\mathcal{Q}}=|\{(i,k): a_{i,k}=s,k\in\mathcal{Q}\}|,\notag
\end{IEEEeqnarray}
and $\mathcal{I}^{\mathcal{Q}}$ be the set of symbols occuring only once in $\bm A^\mathcal{Q}$:
\begin{IEEEeqnarray}{c}
\mathcal{I}^{\mathcal{Q}}\triangleq\{s\in[S]:g_s^{\mathcal{Q}}=1\}.\notag
\end{IEEEeqnarray}
 The symbols in $\mathcal{I}^{\mathcal{Q}}$ are partitioned into
 $\sys{Q}$ subsets $\{\mathcal{I}^{\mathcal{Q}}_k:k\in\mathcal{Q}\}$ as
 follows.  For each $s\in\mathcal{I}^{\mathcal{Q}}$, let $(i,j)$ be the
 unique pair in $[\sys{F}]\times \mathcal{Q}$ such that $a_{i,j}=s$.
 Since the number of $``*"$ symbols in the $i$-th row of $\bm A$ is equal or larger than $\sys{K}-\sys{Q}+1$ by the assumption $\tau\geq \sys{K}-\sys{Q}+1$, there exists at least one $k\in\mathcal{Q}$ such that $a_{i,k}=*$. Arbitrarily choose such a $k$ and assign $s$ into $\mathcal{I}^{\mathcal{Q}}_k$.

  Let $\mathcal{A}_k^\mathcal{Q}$ denote the set of ordinary symbols in column $k$ occurring at least twice:
 \begin{IEEEeqnarray}{c}
 \mathcal{A}_k^\mathcal{Q}\triangleq\{s\in[\sys{S}]: a_{i,k}=s~\mbox{for some}~i\in[\sys{F}]\}\backslash\mathcal{I}^{\mathcal{Q}}, \quad k\in\mathcal{Q}.\label{eqn:AkQ}
 \end{IEEEeqnarray}

Pick any uniform assignment of reduce functions $\mathbf{D}^{\mathcal{Q}}=\{\mathcal{D}^{\mathcal{Q}}_k\}_{k\in\mathcal{Q}}$. Let $\mathcal{U}_{ i,j}^{\mathcal{Q}}$ denote the set of  IVAs for \qifa{node $j$} computed from the files in $\mathcal{W}_{ i}$, i.e.,
\begin{IEEEeqnarray}{c}
\mathcal{U}_{i, j}^{\mathcal{Q}}\triangleq\left\{v_{d,n}:\,d\in\mathcal{D}^\mathcal{Q}_j,\, w_n\in\mathcal{W}_{ i}\right\},\quad (i,j)\in[\sys{F}]\times\mathcal{Q}.\notag
\end{IEEEeqnarray}

\subsubsection{Shuffle  Phase}
  Node $k$ multicasts the  signal
	\begin{IEEEeqnarray}{c}	X_k^{\mathcal{Q}}=\big\{X_{k,s}^\mathcal{Q}:s\in\mathcal{I}_k^{\mathcal{Q}}\cup\mathcal{A}_k^{\mathcal{Q}}\big\},\notag
	\end{IEEEeqnarray}
where the signals $X_{k,s}^\mathcal{Q}$  are created as described in the following,  \mw{depending on whether $s\in\mathcal{I}_k$ or $s \in \mathcal{A}_k^{\mathcal{Q}}$.}
For all $s\in\mathcal{I}_k$, set
\begin{IEEEeqnarray}{c}
X_{k,s}^{\mathcal{Q}}\triangleq \mathcal{U}_{i,j}^{\mathcal{Q}}, \mw{\qquad s\in\mathcal{I}_k,}\label{Xsk:2}
\end{IEEEeqnarray}
where $(i,j)$ is the unique index in $[\sys{F}]\times\mathcal{Q}$ such that $a_{i,j}=s$.

To describe the signal $X_{k,s}^\mathcal{Q}$ for \qifa{$s\in\mathcal{A}_k^{\mathcal{Q}}$}, we first describe a partition of the IVA $\mathcal{U}_{i,j}^{\mathcal{Q}}$ for each pair $(i,j)\in[\sys{F}]\times\mathcal{Q}$ such that  \qifa{$a_{i,j}\in\mathcal{A}_j^{\mathcal{Q}}$}.
 Let $s'=a_{i,j}$, then $g_{s'}^{\mathcal{Q}}\geq 2$.
  Let $(l_1,j_1),(l_2,j_2),\ldots,(l_{g_s^{\mathcal{Q}}-1},j_{g_s^{\mathcal{Q}}-1})\in[\sys{F}]\times\mathcal{Q}$ indicate all the other $g_{s'}^{\mathcal{Q}}-1$ occurrences of the ordinary symbol $s'$ in $\bm A^{\mathcal{Q}}$, i.e.,
\begin{IEEEeqnarray}{c}
a_{l_1,j_1}=a_{l_2,j_2}=\ldots=a_{l_{g_s^{\mathcal{Q}}-1},j_{g_s^{\mathcal{Q}}-1}}=s'.\notag
\end{IEEEeqnarray}
Partition \mw{the set of IVAs} $\mathcal{U}_{i,j}^{\mathcal{Q}}$ into $g_{s'}^{\mathcal{Q}}-1$  \mw{subsets} of equal size and denote these \mw{subsets} by $\mathcal{U}_{i,j}^{\mathcal{Q},j_1},\mathcal{U}_{i,j}^{\mathcal{Q},j_2},\ldots,\mathcal{U}_{i,j}^{\mathcal{Q},j_{g_{s}^\mathcal{Q}-1}}$:
\begin{IEEEeqnarray}{c}
\mathcal{U}_{i,j}^{\mathcal{Q}}=\left\{\mathcal{U}_{i,j}^{\mathcal{Q},j_1},\mathcal{U}_{i,j}^{\mathcal{Q},j_2},\ldots,\mathcal{U}_{i,j}^{\mathcal{Q},j_{g_{s}^\mathcal{Q}-1}}\right\}.\label{Uij:partition}
\end{IEEEeqnarray}

For all $s\in\mathcal{A}_k^{\mathcal{Q}}$, set
	\begin{IEEEeqnarray}{c}
		X_{k,s}^\mathcal{Q}\triangleq\bigoplus_{\substack{(i,j)\in[\sys{F}]\times (\mathcal{Q}\backslash\{k\})  \colon \\[1.1ex] a_{i,j}=s}} \mathcal{U}_{i,j}^{\mathcal{Q},k}, \qquad \mw{s \in \mathcal{A}_k^{\mathcal{Q}}}. \label{Xsk}
	\end{IEEEeqnarray}

	\subsubsection{Reduce Phase} Node $k$  \qifa{computes} all IVAs in
	\begin{IEEEeqnarray}{c}
	\bigcup_{i\in[\sys{F}]} \mathcal{U}_{i,k}^{\mathcal{Q}}.\notag
	\end{IEEEeqnarray}
	In the map phase, node $k$ has already computed all IVAs in $\big\{\mathcal{U}_{i,k}^{\mathcal{Q}}:a_{i,k}=*\big\}$. It thus remains to compute all IVAs in
	\begin{IEEEeqnarray}{c}
		\bigcup_{\substack{i\in[\sys{F}]\colon  a_{i,k}\neq *}} \mathcal{U}_{i,k}^{\mathcal{Q}}.\notag
	\end{IEEEeqnarray}
	Fix an arbitrary $i\in[\sys{F}]$ such that $a_{i,k}\neq *$, and set $s = a_{i,k}$. If $s\in\mathcal{A}_k^{\mathcal{Q}}$,
	each subblock $\mathcal{U}_{i,k}^{\mathcal{Q},j}$ in \eqref{Uij:partition} can be
	restored by node~$k$ from the signal \qifa{$X_{j,s}^\mathcal{Q}$} sent by node $j$ (see \eqref{Xsk}):
\begin{IEEEeqnarray}{c}
\mathcal{U}_{i,k}^{\mathcal{Q},j}=\mathcal{U}_{l_1,j_1}^{\mathcal{Q},j}\; \oplus\; \mathcal{U}_{l_2,j_2}^{\mathcal{Q},j}\; \oplus\; \ldots\; \oplus\; \mathcal{U}_{l_{g_s^{\mathcal{Q}}-2},j_{g_s^{\mathcal{Q}}-2}}^{\mathcal{Q},j}\; \oplus X_{j,s}^{\mathcal{Q}},\label{decode:Uik}
\end{IEEEeqnarray}
where $(l_t,j_t)$ ($t\in[g_s^{\mathcal{Q}}-2]$) \mw{indicate the other $g_s^{\mathcal{Q}}-2$ occurrences of the symbol $s$ in ${\bm A}^{\mathbf{Q}}$, i.e.,  $a_{l_t,j_t}=s$.} Notice that \mw{the  sub-IVAs on the right-hand side of \eqref{decode:Uik}} have been computed by node $k$ during the map phase, because by the PDA properties, $a_{l_t,j_t}=a_{i,k}=s$ and \qifa{$j_t\neq k$ imply that $l_t\neq i$ and $a_{l_t,k}=*$.}
Therefore, $\mathcal{U}_{i,k}^{\mathcal{Q},j}$ can be decoded from \eqref{decode:Uik}.

 If $s\notin\mathcal{A}_k^{\mathcal{Q}}$, then $s\in\mathcal{I}^{\mathcal{Q}}$ by \eqref{eqn:AkQ}. There exists thus an index $j\in\mathcal{Q}\backslash\{k\}$ such that $s\in\mathcal{I}_j$ and therefore, by  \eqref{Xsk:2}, the subset \qifa{$\mathcal{U}_{i,k}^{\mathcal{Q}}$} can be recovered from the signal $X_{j,s}^\mathcal{Q}$ sent by node $j$.

	\begin{remark}\label{remark:PDAmap}
  It is worth pointing out that the storage design
  $\{\mathcal{M}_k\}_{k=1}^\sys{K}$ only depends on the positions of the
  $``*"$ symbols in $\bm A$, but not on the parameter $\sys{Q}$ (See \eqref{PDA:Mk}). This indicates that, in practice the map phase can  be carried out even without knowing how many nodes will be participating in the shuffle and reduce phases.
  \end{remark}

\subsection{Performance Analysis}\label{subsec:performancePDA}

We have analyzed the performances of storage and communication loads in
the no-stragglers setup in \cite{Yan2018SCC}. For the scheme in the
preceding subsection, the analysis of storage load follows the same
lines as in \cite{Yan2018SCC}.  When computing the  communication load defined in \eqref{communication_load}, we have to average over all realizations of the active set $\mathbf{Q}$.
\subsubsection{Storage Load} Since the Comp-PDA $\bm A$ contains
$\sys{T}$  $``*"$ symbols, and each $``*"$ symbol indicates that a batch of $\eta=\frac{\sys{N}}{\sys{F}}$ files is stored at a given node, see \eqref{PDA:Mk}, the storage load of the proposed scheme is:
\begin{IEEEeqnarray}{c}
r=\frac{\sum_{k=1}^{\mathsf{K}}|\mathcal{M}_k|}{\mathsf{N}}=\frac{\sys{T}\cdot\eta}{\sys{N}}=\frac{\sys{T}}{\sys{F}}.\notag
\end{IEEEeqnarray}

\subsubsection{Communication Load} \mw{We first analyze the length of the signals sent for a given realization of the active set $\mathbf{Q}=\mathcal{Q}$.}
  For any $s\in[\mathsf{S}]$, let $g_s$ be the occurrence of $s$ in $\bm A$, and $g_s^{\mathcal{Q}}$ be the occurrence of $s$ in the columns in $\mathcal{Q}$. By  \eqref{Xsk:2} and \eqref{Xsk}, the length of the signals associated to symbol $s$ is
\begin{IEEEeqnarray}{c}
l_s^{\mathcal{Q}}=\left\{\begin{array}{ll}
0,&\text{if}~g_s^{\mathcal{Q}}=0\\
                          \mathsf{\frac{VND}{FQ}},  &\text{if~}g_{s}^{\mathcal{Q}}=1  \\
                          \frac{g_{s}^{\mathcal{Q}}}{g_s^{\mathcal{Q}}-1}\cdot \mathsf{\frac{VND}{FQ}} , & \text{if~}g_{s}^{\mathcal{Q}}\geq2
                         \end{array}
\right.,\label{lsQ}
\end{IEEEeqnarray}
when $\mathcal{Q}$ is the active set. The total length of all the
signals is thus
\begin{IEEEeqnarray}{rCl}
\sum_{k\in\mathcal{Q}}|X_k^{\mathcal{Q}}|&=&\sum_{k\in\mathcal{Q}}\sum_{s\in[\sys{S}]:g_s^{\mathcal{Q}}>0}|X_{k,s}^{\mathcal{Q}}|\notag\\
&=&\sum_{s\in[\sys{S}]:g_s^{\mathcal{Q}}>0}\sum_{k\in\mathcal{Q}}|X_{k,s}^{\mathcal{Q}}|\notag\\
&=&\sum_{s\in[\sys{S}]}l_s^{\mathcal{Q}}. \label{eqn:totallength}
\end{IEEEeqnarray}

\mw{We now compute  the communication load as defined in \eqref{communication_load} where we have to average over all realizations of the active set $\mathbf{Q}$:} 
\begin{IEEEeqnarray}{rCl}
&&L_{\bm A}\notag\\
&=& \mathbf{E}\left[\frac{\sum_{k\in\mathbf{Q}}|X_k^{\mathbf{Q}}|}{\sys{NDV}}\right]\notag\\
&=&\sys{\frac{1}{NDV}}\cdot\frac{1}{|\mathbf{\Omega}_{\sys{K}}^{\sys{Q}}|}\cdot\sum_{\mathcal{Q}\in\mathbf{\Omega}_{\sys{K}}^{\sys{Q}}} \sum_{k\in\mathcal{Q}}|X_k^{\mathcal{Q}}|\notag\\
&\overset{(a)}{=}&\frac{1}{\sys{NDV}}\cdot\frac{1}{C_\sys{K}^\sys{Q}}\cdot\sum_{\mathcal{Q}\in\mathbf{\Omega}_{\sys{K}}^{\sys{Q}}} \sum_{s\in[\sys{S}]}l_s^{\mathcal{Q}}\notag\\
&\overset{(b)}{=}&\sys{\frac{1}{NDV}}\cdot\frac{1}{C_\sys{K}^\sys{Q}}\cdot\sum_{s\in[\sys{S}]}\sum_{\mathcal{Q}\in\mathbf{\Omega}_{\sys{K}}^{\sys{Q}}} l_s^{\mathcal{Q}}\cdot\left(\sum_{l=0}^\sys{Q}\mathbbm{1}(g_s^{\mathcal{Q}}=l)\right)\cdot\left(\sum_{g=1}^\sys{K} \mathbbm{1}(g_s=g)\right)\notag\\
&=&\sys{\frac{1}{NDV}}\cdot\frac{1}{C_\sys{K}^\sys{Q}}\cdot\sum_{g=1}^\sys{K} \sum_{s\in[\sys{S}]}\sum_{l=0}^\sys{Q}\sum_{\mathcal{Q}\in\mathbf{\Omega}_{\sys{K}}^{\sys{Q}}}l_s^{\mathcal{Q}}\cdot\mathbbm{1}(g_s^{\mathcal{Q}}=l)\cdot\mathbbm{1}(g_s=g)\notag\\
&\overset{(c)}{=}&\sys{\frac{1}{NDV}}\cdot\frac{1}{C_\sys{K}^\sys{Q}}\cdot\sum_{g=1}^\sys{K} \sum_{s\in[\sys{S}]} \left(\sum_{\mathcal{Q}\in\mathbf{\Omega}_{\sys{K}}^{\sys{Q}}}\sys{\frac{NDV}{FQ}}\cdot\mathbbm{1}(g_s^{\mathcal{Q}}=1)+\sum_{l=2}^\sys{Q}\sum_{\mathcal{Q}\in\mathbf{\Omega}_{\sys{K}}^{\sys{Q}}}\frac{l\sys{NDV}}{(l-1)\sys{FQ}}\cdot\mathbbm{1}(g_s^{\mathcal{Q}}=l)\right)\cdot\mathbbm{1}(g_s=g)\IEEEeqnarraynumspace\notag\\
&=&\sys{\frac{1}{FQ}}\cdot\frac{1}{C_\sys{K}^\sys{Q}}\cdot\sum_{g=1}^\sys{K} \sum_{s\in[\sys{S}]} \left[\sum_{\mathcal{Q}\in\mathbf{\Omega}_{\sys{K}}^{\sys{Q}}}\mathbbm{1}(g_s^{\mathcal{Q}}=1)+\sum_{l=2}^\sys{Q}\frac{l}{l-1}\cdot\left(\sum_{\mathcal{Q}\in\mathbf{\Omega}_{\sys{K}}^{\sys{Q}}}\mathbbm{1}(g_s^{\mathcal{Q}}=l)\right)\right]\cdot\mathbbm{1}(g_s=g)\IEEEeqnarraynumspace\notag\\
&\overset{(d)}{=}&\sys{\frac{1}{FK}}\cdot\frac{1}{C_{\sys{K}-1}^{\sys{Q}-1}}\cdot\sum_{g=1}^\sys{K} \sum_{s\in[\sys{S}]}\left(C_{g}^1\cdot C_{\sys{K}-g}^{\sys{Q}-1}+\sum_{l=2}^\sys{Q}\frac{l}{l-1}\cdot C_{g}^l\cdot C_{\sys{K}-g}^{\sys{Q}-l}\right)\cdot\mathbbm{1}(g_s=g)\notag\\
&=&\frac{1}{\sys{FK}}\cdot\frac{1}{C_{\sys{K}-1}^{\sys{Q}-1}}\cdot\sum_{g=1}^\sys{K} \left(C_{g}^1\cdot C_{\sys{K}-g}^{\sys{Q}-1}+\sum_{l=2}^\sys{Q}\frac{l}{l-1}\cdot C_{g}^l\cdot C_{\sys{K}-g}^{\sys{Q}-l}\right)\cdot\sum_{s\in[\sys{S}]}\mathbbm{1}(g_s=g)\notag\\
&\overset{(e)}{=}&\frac{1}{\sys{FK}}\cdot\frac{1}{C_{\sys{K}-1}^{\sys{Q}-1}}\cdot\sum_{g=1}^\sys{K} \sys{S}_g \left(C_{g}^1\cdot C_{\sys{K}-g}^{Q-1}+\sum_{l=2}^\sys{Q}\frac{l}{l-1}\cdot C_{g}^l\cdot C_{\sys{K}-g}^{\sys{Q}-l}\right)\notag\\
&\overset{(f)}{=}&\frac{1}{\sys{FK}}\cdot\frac{1}{C_{\sys{K}-1}^{\sys{Q}-1}}\cdot\sum_{g=1}^\sys{K} \sys{S}_g\left(g\cdot C_{\sys{K}-g}^{\sys{Q}-1}+\sum_{l=2}^\sys{Q}\frac{g}{l-1}\cdot C_{g-1}^{l-1}\cdot C_{\sys{K}-g}^{\sys{Q}-l}\right)\notag\\
&=&\frac{1}{\sys{FK}}\cdot\frac{1}{C_{\sys{K}-1}^{\sys{Q}-1}}\cdot\sum_{g=1}^\sys{K} \sys{S}_gg\left(C_{\sys{K}-g}^{\sys{Q}-1}+\sum_{l=1}^{\sys{Q}-1}\frac{1}{l}\cdot C_{g-1}^l\cdot C_{\sys{K}-g}^{\sys{Q}-l-1}\right)\label{sum:f}\\
&\overset{(g)}{=}&\frac{1}{\sys{FK}}\cdot\frac{1}{C_{\sys{K}-1}^{\sys{Q}-1}}\cdot\sum_{g=1}^\sys{K} \sys{S}_gg\left(C_{\sys{K}-g}^{\sys{Q}-1}+\sum_{l=\max\{1,g-\sys{K+Q}-1\}}^{\min\{g,\sys{Q}\}-1}\frac{1}{l}\cdot C_{g-1}^l\cdot C_{\sys{K}-g}^{\sys{Q}-l-1}\right)\notag\\
&=&\frac{\sys{FK}-\sys{T}}{\sys{FK}}\cdot \frac{1}{C_{\sys{K}-1}^{\sys{Q}-1}}\cdot\sum_{g=1}^\sys{K} \frac{\sys{S}_gg}{\sys{FK}-\sys{T}}\cdot\left(C_{\sys{K}-g}^{\sys{Q}-1}+\sum_{l=\max\{1,g-\sys{K+Q}-1\}}^{\min\{g,\sys{Q}\}-1}\frac{1}{l}\cdot C_{g-1}^l\cdot C_{\sys{K}-g}^{\sys{Q}-l-1}\right)\notag\\
&\overset{(h)}{=}&\left(1-\frac{\sys{T}}{\sys{FK}}\right)\cdot\frac{1}{C_{\sys{K}-1}^{\sys{Q}-1}}\cdot\sum_{g=1}^\sys{K}\theta_g\left(C_{\sys{K}-g}^{\sys{Q}-1}+\sum_{l=\max\{1,g-\sys{K+Q}-1\}}^{\min\{g,\sys{Q}\}-1}\frac{1}{l}\cdot C_{g-1}^l\cdot C_{\sys{K}-g}^{\sys{Q}-l-1}\right),\notag
\end{IEEEeqnarray}
where $(a)$ holds by \eqref{eqn:totallength};
 $(b)$ holds since for each $s\in[\sys{S}]$,
$
\sum_{l=0}^\sys{Q}\mathbbm{1}(g_s^{\mathcal{Q}}=l)=1$ and $\sum_{g=1}^\sys{K} \mathbbm{1}(g_s=g)=1
$;
$(c)$ follows from \eqref{lsQ}; and $(d)$ holds since for each symbol occurring $g$ times, it has occurrence $l$ in exact $C_{g}^l\cdot C_{\sys{K}-g}^{\sys{Q}-l}$ subsets of $\mathcal{K}$ with cardinality $\sys{Q}$; in $(e)$, we defined $\sys{S}_g$ to be the number of ordinary symbols occurring $g$ times for each $g\in[\sys{K}]$; in $(f)$, we used the equality $C_g^l=\frac{g}{l}\cdot C_{g-1}^{l-1}$; in $(g)$, we eliminated the indices of zero terms in the summation of \eqref{sum:f}; and $(h)$ follows from the definition of symbol frequencies.

\subsection{File Complexity of the Proposed Schemes}\label{subsec:complexity}
The analysis of file complexity is similar to the no-straggler setup in \cite{Yan2018SCC}.
 The  files are partitioned into $\sys{F}$ batches so that  each batch
contains $\eta=\frac{\sys{N}}{\sys{F}}>0$ files. It is assumed that  $\eta$ is a positive integer. The smallest number of files $\sys{N}$ where this assumption can be met is $\sys{F}$. Therefore,  the file complexity of the scheme is $\sys{F}$.

\section{The Fundamental Storage-Communication Tradeoff (Proof of Theorem \ref{thm:limits})}\label{sec:SCtradeoff}

By Corollary \ref{corollary:MANPDA}, the SC pair $\big(r,L_{\sys{K},\sys{Q}}^*(r)\big)$, $r\in[\sys{K}-\sys{Q}+1:\sys{K}]$ can be
achieved by the MAN-PDA $\bm P_r$.  For a general non-integer $r\in[\sys{K}-\sys{Q}+1,\sys{K}]$, the lower convex envelope of these points can be achieved by memory- and time- sharing. It remains to prove the converse in Theorem \ref{thm:limits}.

Let $Z_\sys{K}^\sys{Q}(x)$ be a piecewise linear function connecting the
points \qifa{$\big(u,Z_\sys{K}^\sys{Q}(u)\big)$} sequentially over the interval $[\sys{K}-\sys{Q}+1,\sys{K}]$
  with
\begin{IEEEeqnarray}{rCl}
Z_\sys{K}^\sys{Q}(u)&\triangleq&\sum_{l=u+\sys{Q}-\sys{K}}^{\min\{u,\sys{Q}\}}\frac{\sys{Q}-l}{\sys{Q}l}
C_u^l C_{\sys{K}-u}^{\sys{Q}-l},\quad
u\in[\sys{K}-\sys{Q}+1:\sys{K}].\label{eqn:ZKQ}
\end{IEEEeqnarray}
We shall need the following lemma, proved in Appendix \ref{sec:claimZ}.
\begin{lemma}\label{lemma:Z}
The sequence $Z_\sys{K}^\sys{Q}(u)$ is strictly convex and decreasing
for $u\in[\sys{K}-\sys{Q}+1:\sys{K}]$. And the function
$Z_\sys{K}^\sys{Q}(x)$ is convex and decreasing over $[\sys{K}-\sys{Q}+1,\sys{K}]$.
\end{lemma}

Let $\mathcal{M}\triangleq\{\mathcal{M}_k\}_{k=1}^\sys{K}$ be a storage design and $(r,L)$ be a SC pair achieved based on $\{\mathcal{M}_k\}_{k=1}^\sys{K}$.
For each $u\in[\sys{K-Q}+1:\sys{K}]$, define
\begin{IEEEeqnarray}{c}
a_{\mathcal{M},u}\triangleq\sum_{\mathcal{I}\subseteq\mathcal{K}:|\mathcal{I}|=u}\left|\left(\mathop\cap\limits_{k\in\mathcal{I}}\mathcal{M}_k\right)\bigg\backslash\left(\mathop\cup\limits_{\bar k\in\mathcal{K}\backslash\mathcal{I}}\mathcal{M}_{ \bar k}\right)\right|,\label{def:aMs}
\end{IEEEeqnarray}
i.e., $a_{\mathcal{M},u}$ is the number of files stored $u$ times across all the nodes. Then by definition, $a_{\mathcal{M},u}$ satisfies
\begin{IEEEeqnarray}{rCl}
a_{\mathcal{M},u}&\geq&0,\notag\\
\sum_{u=\sys{K-Q}+1}^{\sys{K}}a_{\mathcal{M},u}&=&\sys{N},\notag\\
\sum_{u=\sys{K-Q}+1}^{\sys{K}}ua_{\mathcal{M},u}&=&\overline{r}\sys{N}.\label{constraint:s}
\end{IEEEeqnarray}

For any  $\mathcal{Q}\in\mathbf{\Omega}_{\sys{K}}^{\sys{Q}}$  and any
$l\in[\sys{Q}]$, define
\begin{IEEEeqnarray}{c}
b_{\mathcal{M},l}^{\mathcal{Q}}\triangleq\sum_{\mathcal{I}\subseteq\mathcal{Q}:|\mathcal{I}|=l}\left|\left(\mathop\cap\limits_{k\in\mathcal{I}}\mathcal{M}_k\right)\bigg\backslash\left(\mathop\cup\limits_{\bar{k}\in\mathcal{Q}\backslash\mathcal{I}}\mathcal{M}_{\bar k}\right)\right|,\notag
\end{IEEEeqnarray}
i.e., $b_{\mathcal{M},l}^{\mathcal{Q}}$ is the number of files stored
exactly $l$ times in the nodes of \mw{set $\mathcal{Q}$}. Since any file that is
stored $u$ times across all the nodes has $l$ occurrences in exactly
$C_u^l\cdot C_{\sys{K}-u}^{\sys{Q}-l}$ subsets \mw{
$\mathcal{Q}$ of size} $\mathfrak{\sys{Q}}$, i.e.,
\begin{IEEEeqnarray}{rCl}
&&\sum_{\mathcal{Q}\in\mathbf{\Omega}_{\sys{K}}^{\sys{Q}}}\mathbbm{1}(w_n~\textnormal{is stored at exactly $l$ nodes of $\mathcal{Q}$})\notag\\
&=&\sum_{u=\max\{l,\sys{K}-\sys{Q}+1\}}^{\sys{K}-\sys{Q}+l}\mathbbm{1}(w_n~\textnormal{is stored at exactly $u$ nodes of $\mathcal{K}$})\cdot C_{u}^l\cdot C_{\sys{K}-u}^{\sys{Q}-l},\quad \forall\,n\in[\sys{N}].\notag
\end{IEEEeqnarray}
Summing over $n\in[\sys{N}]$, we obtain
\begin{IEEEeqnarray}{c}
\sum_{\mathcal{Q}\in\mathbf{\Omega}_\sys{K}^{\sys{Q}}}b_{\mathcal{M},l}^{\mathcal{Q}}=\sum_{u=\max\{l,\sys{K-Q}+1\}}^{\sys{K-Q}+l}a_{\mathcal{M},u}\cdot C_u^l\cdot C_{\sys{K}-u}^{\sys{Q}-l}.\label{eqn:sumequality}
\end{IEEEeqnarray}

\mw{We can now apply the result in \cite[Lemma 1]{Li2018Tradeoff} of the system without stragglers, to  lower bound
the  communication load for any realization of the active set $\mathbf{Q}=\mathcal{Q}$:}
\begin{IEEEeqnarray}{c}
\mw{\frac{\sum_{k\in\mathcal{Q}}\,|\,X_k^\mathcal{Q}|}{\sys{NDV}}}\geq\sum_{l=1}^{\sys{Q}}\frac{b_{\mathcal{M},l}^{\mathcal{Q}}}{\sys{N}}\frac{\sys{Q}-l}{\sys{Q}l}.\notag
\end{IEEEeqnarray}

\mw{The average communication load over the random realization of the active set $\mathbf{Q}$ is then obtained as:}
\begin{IEEEeqnarray}{rCl}
\overline{L}&=& \mathbf{E}_{\mathbf{Q}}\left[\frac{\sum_{k\in\mathbf{Q}}\,|\,X_k^\mathbf{Q}|}{\sys{NDV}}\right]\notag\\
&=&\sum_{\mathcal{Q}\in\mathbf{\Omega}_{\sys{K}}^{\sys{Q}}} \frac{\sum_{k\in\mathcal{Q}}\,|\,X_k^\mathcal{Q}|}{\sys{NDV}}\,\cdot\textnormal{Pr}\{\mathbf{Q}=\mathcal{Q}\}\IEEEeqnarraynumspace\notag\\
&\geq&\frac{1}{C_{\sys{K}}^\sys{Q}}\sum_{\mathcal{Q}\in\mathbf{\Omega}_{\sys{K}}^{\sys{Q}}}\sum_{l=1}^{\sys{Q}}\frac{b_{\mathcal{M},l}^{\mathcal{Q}}}{\sys{N}}\frac{\sys{Q}-l}{\sys{Q}l}\notag\\
&=&\frac{1}{C_\sys{K}^\sys{Q}}\sum_{l=1}^\sys{Q}\left(\sum_{\mathcal{Q}\in\mathbf{\Omega}_{\sys{K}}^{\sys{Q}}}\frac{b_{\mathcal{M},l}^{\mathcal{Q}}}{\sys{N}}\right)\frac{\sys{Q}-l}{\sys{Q}l}\notag\\
&\overset{(a)}{=}& \frac{1}{C_\sys{K}^\sys{Q}}\sum_{l=1}^{\sys{Q}}\left(\sum_{u=\max\{l,\sys{K}-\sys{Q}+1\}}^{\sys{K}-\sys{Q}+l}\frac{a_{\mathcal{M},u}}{\sys{N}} C_u^l C_{\sys{K}-u}^{\sys{Q}-l}\right)\frac{\sys{Q}-l}{\sys{Q}l}\IEEEeqnarraynumspace\label{eq:double_sum}\\
&\overset{(b)}{=}& \frac{1}{C_\sys{K}^\sys{Q}} \sum_{u=\sys{K}-\sys{Q}+1}^\sys{K}\frac{a_{\mathcal{M},u}}{\sys{N}}\sum_{l=u+\sys{Q}-\sys{K}}^{\min\{u,\sys{Q}\}}{C_u^l C_{\sys{K}-u}^{\sys{Q}-l}} \frac{\sys{Q}-l}{\sys{Q}l}\notag\\
&\overset{(c)}{=}&\frac{1}{C_\sys{K}^\sys{Q}}\sum_{u=\sys{K}-\sys{Q}+1}^\sys{K}\frac{a_{\mathcal{M},u}}{\sys{N}}\cdot Z_\sys{K}^\sys{Q}(u)\notag\\
&\overset{(d)}{\geq}&\frac{1}{C_\sys{K}^\sys{Q}}\cdot Z_{\sys{K}}^\sys{Q}\left(\sum_{u=\sys{K}-\sys{Q}+1}^\sys{K}\frac{ua_{\mathcal{M},u}}{\sys{N}}\right)\label{convex:inequality}\\
&\overset{(e)}{=}&\frac{Z_\sys{K}^\sys{Q}(\overline{r})}{C_\sys{K}^\sys{Q}}\notag\\
&\overset{(f)}{\geq}&\frac{Z_\sys{K}^\sys{Q}(r+\epsilon)}{C_\sys{K}^\sys{Q}},\notag
\end{IEEEeqnarray}
where $(a)$ follows from \eqref{eqn:sumequality}; \mw{$(b)$ holds because the inner summation in \eqref{eq:double_sum} only includes summation indices $u\in[\sys{K}-\sys{Q}:\sys{K}]$ and it includes  the summation index $u\in\{ \sys{K}-\sys{Q}+1, \ldots, \sys{K}\}$ if, and only if,  the outer summation index  $l$ satisfies $l \leq u$ and $l \geq u+\sys{Q}-\sys{K}$;
}$(c)$ follows from \eqref{eqn:ZKQ}; $(d)$ follows from Lemma \ref{lemma:Z}; $(e)$ follows from \eqref{constraint:s}; and $(f)$ follows from the fact $\overline{r}\leq r+\epsilon$.
Since $\epsilon$ can be \sheng{arbitrarily} close to zero, we conclude
\begin{IEEEeqnarray}{c}
L\geq\frac{Z_\sys{K}^\sys{Q}(r)}{C_\sys{K}^\sys{Q}}.\notag
\end{IEEEeqnarray}
In particular, when $r\in[\sys{K}-\sys{Q}+1:\sys{K}]$, by \eqref{eqn:ZKQ},
\begin{IEEEeqnarray}{rCl}
L&\geq& \sum_{l=r+\sys{Q}-\sys{K}}^{\min\{r,\sys{Q}\}}\frac{\sys{Q}-l}{\sys{Q}l}\frac{C_{r}^l C_{\sys{K}-r}^{\sys{Q}-l}}{C_\sys{K}^\sys{Q}}\notag\\
&\overset{(a)}{=}& \sum_{l=r+\sys{Q}-\sys{K}}^{\min\{r,\sys{Q}-1\}}\frac{\sys{Q}-l}{\sys{Q}l}\frac{C_{r}^l C_{\sys{K}-r}^{\sys{Q}-l}}{C_\sys{K}^\sys{Q}}\notag\\
&=&\sum_{l=r+\mathsf{Q-K}}^{\min\{r,\mathsf{Q}-1\}}\frac{\mathsf{Q}-l}{\mathsf{Q}l}\cdot\frac{\frac{\mathsf{Q}!}{l!(\mathsf{Q}-l)!}\cdot\frac{(\mathsf{K-Q})!}{(r-l)!(\mathsf{K-Q}-r+l)!}}{\frac{\mathsf{K}!}{r!(\mathsf{K}-r)!}}\notag\\
&=&\left(1-\frac{r}{\mathsf{K}}\right)\cdot\sum_{l=r+\mathsf{Q-K}}^{\min\{r,\mathsf{Q}-1\}}\frac{1}{l}\cdot\frac{\frac{r!}{l!(r-l)!}\cdot\frac{(\mathsf{K}-r-1)!}{(\mathsf{Q}-l-1)!(\mathsf{K}-r-\mathsf{Q}+l)!}}{\frac{(\mathsf{K}-1)!}{(\mathsf{Q}-1)!(\mathsf{K-Q})!}}\notag\\
&\overset{(b)}{=}&\left(1-\frac{r}{\sys{K}}\right)\sum_{l=r+\sys{Q}-\sys{K}}^{\min\{r,\sys{Q}-1\}}\frac{1}{l}\frac{C_{r}^l C_{\sys{K}-r-1}^{\sys{Q}-l-1}}{C_{\sys{K}-1}^{\sys{Q}-1}},\notag
\end{IEEEeqnarray}
where $(a)$ holds since for $l=\sys{Q}$, the term in the summation is
zero. This establishes the desired converse proof.



\section{Optimality of File Complexity (Proof of Theorem \ref{thm:optimalF})}\label{sec:proofofthmF}
In order to prove Theorem \ref{thm:optimalF}, we need to first derive several lemmas.

\subsection{Preliminaries}

\begin{lemma}\label{cor:rtime} If a coded computing scheme achieves
  the fundamental SC tradeoff pair $\left(r,L_{\sys{K},\sys{Q}}^*(r)\right)$ for any integer
  $r\in[\sys{K}-\sys{Q}+1:\sys{K}]$, then each file is stored exactly
  $r$ times across the nodes.
\end{lemma}
\begin{IEEEproof} According to Lemma~\ref{lemma:Z}, the sequence
  $\big\{Z_\sys{K}^\sys{Q}(u)\big\}_{u=\sys{K}-\sys{Q}+1}^\sys{K}$ is
  strictly convex.  Thus for the integer $r =
  \sum_{u=\sys{K}-\sys{Q}+1}^\sys{K}\frac{ua_{\mathcal{M},u}}{\sys{N}}$,
  the equality in \eqref{convex:inequality} holds if, and only if,
\begin{IEEEeqnarray}{rCl}
\frac{a_{\mathcal{M},r}}{\sys{N}}&=&1,\notag\\
\frac{a_{\mathcal{M},u}}{\sys{N}}&=&0,\quad u\in [\sys{K}-\sys{Q}+1:\sys{K}]\backslash\{r\}.\notag
\end{IEEEeqnarray}
Therefore, by definition of $a_{\mathcal{M},u}$ in \eqref{def:aMs}, this indicates that each file is stored exactly $r$ times across the system.
\end{IEEEproof}

\begin{lemma}\label{lemma:lowerbound:F}
In a $g$-regular $\sys{(K,F,T,S)}$ PDA, s.t., $\sys{K}\geq
g\geq 2$, if there are exactly $g-1$ $``*"$s in each row,
then $\sys{F}\geq C_{\sys{K}}^{g-1}$.
\end{lemma}
\begin{proof}
With Definition \ref{def:PDA} (the definition of PDAs), the conclusion follows directly from
\cite[Lemma 2]{Yan2017PDA}.
\end{proof}

For each $u\in[\sys{K}]$, define
\begin{IEEEeqnarray}{c}
U_\sys{K}^\sys{Q}(u)\triangleq C_{\sys{K}-u}^{\sys{Q}-1}+\sum_{l=\max\{1,u-\sys{K+Q}-1\}}^{\min\{u,\sys{Q}\}-1}\frac{C_{u-1}^l\cdot C_{\sys{K}-u}^{\sys{Q}-l-1}}{l}.\label{def:UKQ}
\end{IEEEeqnarray}

\begin{lemma}\label{lem:UKQ}
When $\sys{Q}\geq 3$, the subsequence $\{U_\sys{K}^\sys{Q}(u)\}_{u=2}^\sys{K}$  strictly decreases with $u\in[2:\sys{K}]$.
\end{lemma}
\begin{IEEEproof}
For each $u \in[2:\sys{K}-1]$, by \eqref{def:UKQ},
\begin{IEEEeqnarray}{rCl}
&&U_\sys{K}^\sys{Q}(u+1)-U_\sys{K}^\sys{Q}(u)\notag\\
&=&-C_{\sys{K}-u-1}^{\sys{Q}-2}+\sum_{l=\max\{1,u-\sys{K}+\sys{Q}\}}^{\min\{u,\sys{Q}-1\}}\frac{C_u^l\cdot C_{\sys{K}-u-1}^{\sys{Q}-l-1}}{l}-\sum_{l=\max\{1,u-\sys{K}+\sys{Q}-1\}}^{\min\{u,\sys{Q}\}-1}\frac{C_{u-1}^l\cdot C_{\sys{K}-u}^{\sys{Q}-l-1}}{l}\notag\\
&\overset{(a)}{=}&-C_{\sys{K}-u-1}^{\sys{Q}-2}+\sum_{l=\max\{1,u-\sys{K}+\sys{Q}\}}^{\min\{u,\sys{Q}-1\}}\frac{(C_{u-1}^l+C_{u-1}^{l-1})\cdot C_{\sys{K}-u-1}^{\sys{Q}-l-1}}{l}\notag\\
&&\qquad\qquad-\sum_{l=\max\{1,u-\sys{K}+\sys{Q}-1\}}^{\min\{u,\sys{Q}\}-1}\frac{C_{u-1}^l\cdot (C_{\sys{K}-u-1}^{\sys{Q}-l-1}+C_{\sys{K}-u-1}^{\sys{Q}-l-2})}{l}\label{sum:CKg}\\
&\overset{(b)}{=}&-C_{\sys{K}-u-1}^{\sys{Q}-2}+\sum_{l=\max\{1,u-\sys{K}+\sys{Q}\}}^{\min\{u,\sys{Q}\}-1}\frac{C_{u-1}^l\cdot C_{\sys{K}-u-1}^{\sys{Q}-l-1}}{l}+\sum_{l=\max\{1,u-\sys{K}+\sys{Q}\}}^{\min\{u,\sys{Q}-1\}}\frac{C_{u-1}^{l-1}\cdot C_{\sys{K}-u-1}^{\sys{Q}-l-1}}{l}\notag\\
&&\qquad\qquad-\sum_{l=\max\{1,u-\sys{K}+\sys{Q}\}}^{\min\{u,\sys{Q}\}-1}\frac{C_{u-1}^l\cdot C_{\sys{K}-u-1}^{\sys{Q}-l-1}}{l}-\sum_{l=\max\{1,u-\sys{K}+\sys{Q}-1\}}^{\min\{u-1,\sys{Q}-2\}}\frac{C_{u-1}^l\cdot C_{\sys{K}-u-1}^{\sys{Q}-l-2}}{l}\notag\\
&=&-C_{\sys{K}-u-1}^{\sys{Q}-2}+\sum_{l=\max\{1,u-\sys{K}+\sys{Q}\}}^{\min\{u,\sys{Q}-1\}}\frac{C_{u-1}^{l-1}\cdot C_{\sys{K}-u-1}^{\sys{Q}-l-1}}{l}-\sum_{l=\max\{1,u-\sys{K}+\sys{Q}-1\}}^{\min\{u-1,\sys{Q}-2\}}\frac{C_{u-1}^l\cdot C_{\sys{K}-u-1}^{\sys{Q}-l-2}}{l}\notag\\
&\overset{(c)}{=}&-C_{\sys{K}-u-1}^{\sys{Q}-2}+\sum_{l=\max\{1,u-\sys{K}+\sys{Q}\}}^{\min\{u,\sys{Q}-1\}}\frac{C_{u-1}^{l-1}\cdot C_{\sys{K}-u-1}^{\sys{Q}-l-1}}{l}-\sum_{l=\max\{2,u-\sys{K}+\sys{Q}\}}^{\min\{u,\sys{Q}-1\}}\frac{C_{u-1}^{l-1}\cdot C_{\sys{K}-u-1}^{\sys{Q}-l-1}}{l-1}\notag\\
&=&-\sum_{l=\max\{2,u-\sys{K}+\sys{Q}\}}^{\min\{u,\sys{Q}-1\}}\frac{C_{u-1}^{l-1}\cdot C_{\sys{K}-u-1}^{\sys{Q}-l-1}}{l(l-1)}\label{eq:negative}\\
&\leq&0,\notag
\end{IEEEeqnarray}
where in $(a)$, we used \eqref{eqn:combination}; in $(b)$, we separate the two summations in \eqref{sum:CKg}  and eliminate the
indices of zero terms in the separated summations; and in $(c)$, we used
the variable change $l'= l + 1$. Moreover, if $u\geq 2$ and $\sys{Q}\geq
3$, from \eqref{eq:negative},
$U_\sys{K}^\sys{Q}(u+1)-U_\sys{K}^\sys{Q}(u)\leq -\frac{C_{u-1}^l
C_{\sys{K}-u-1}^{\sys{Q}-2}}{2}<0$,
i.e., $U_\sys{K}^\sys{Q}(u)$ is strictly decreasing when $u\geq 2$.
\end{IEEEproof}

\subsection{Proof of Theorem \ref{thm:optimalF}}
Define the set
\begin{IEEEeqnarray}{c}
 \mathcal{E}\triangleq \{(\sys{Q},r):\sys{Q}\in[\sys{K}], r\in[\sys{K}-\sys{Q}+1:\sys{K}]\},\notag
\end{IEEEeqnarray}
and partition  it into three subsets
\begin{IEEEeqnarray}{rCl}
\mathcal{E}_1&\triangleq&\{(\sys{Q},r):\sys{Q}\in[\sys{K}],r=\sys{K}\},\notag\\
\mathcal{E}_2&\triangleq&\{(\sys{Q},r):\sys{Q}\in\{2,\sys{K}\},r=\sys{K}-\sys{Q}+1\},\notag\\
\mathcal{E}_3&\triangleq&\{(\sys{Q},r):\sys{Q}\in[3:\sys{K}],r\in[\max\{\sys{K}-\sys{Q}+1,2\}:\sys{K}-1]\}.\notag
\end{IEEEeqnarray}

 Notice that, if $(\sys{Q},r)\in\mathcal{E}_1$, the bound $\sys{F}\geq C_\sys{K}^{\sys{K}}=1$ is trivial. The case $(\sys{Q},r)\in\mathcal{E}_2$ is excluded.  Therefore,
in the rest of the proof, we assume $(\sys{Q},r)\in\mathcal{E}_3$, i.e., $\sys{Q}\in[3:\sys{K}]$ and $r\in[\max\{\sys{K}-\sys{Q}+1,2\}:\sys{K}-1]$.

 Let $\bm{A}$ be a $\sys{(K,F,T,S)}$ Comp-PDA that achieves the optimal
 tradeoff point $(r, L_{\sys{K},\sys{Q}}^*(r))$.
Recall that each row in a Comp-PDA is associated to a file batch, and a
 $``*''$ symbole in that row and column $k$ indicates that the file batch is
 stored at node $k$. According to Lemma~\ref{cor:rtime}, each file is
 stored exactly $r$ times across the nodes, i.e., there are exactly $r$
 $``*''$ symbols in each row of $\bm A$.

 Let $\theta_{g'}$ be the fraction of ordinary entries  occurring  $g'$ times
 in the Comp-PDA $\bm A$, for all $g'\in[\sys{K}]$. Since there are $r$
 $``*''$ symbols in each row,
 from the PDA properties a) and b) in Definition \ref{def:PDA},  any ordinary symbol cannot appear more than $r+1$ times, i.e.,
 \begin{IEEEeqnarray}{c}
 \sys{\theta}_{g'}=0,\quad \forall\; g'\in[r+2:\sys{K}].\label{S:zero}
 \end{IEEEeqnarray}
 Therefore,
\begin{IEEEeqnarray}{c}
\sum_{g'=1}^\sys{r+1}\theta_{g'}=1.\label{sum:theta_g}
\end{IEEEeqnarray}
From \eqref{LA:PDA}, \eqref{def:UKQ}, and \eqref{S:zero}, the communication load of $\bm A$ has the form
\begin{IEEEeqnarray}{rCl}
L_{\bm A}&=&\left(1-\frac{r}{\sys{K}}\right)\cdot\frac{1}{C_{\sys{K}-1}^{\sys{Q}-1}}\cdot\sum_{g'=1}^{r+1}\theta_{g'}\cdot U_\sys{K}^\sys{Q}(g')\notag\\
&\overset{(a)}{\geq}&\left(1-\frac{r}{\sys{K}}\right)\cdot\frac{1}{C_{\sys{K}-1}^{\sys{Q}-1}}\cdot\left(\sum_{g'=1}^{r+1}\theta_{g'}\right)\cdot U_\sys{K}^\sys{Q}(r+1)\label{equality:ahold}\\
&\overset{(b)}{=}&\left(1-\frac{r}{\sys{K}}\right)\cdot\frac{1}{C_{\sys{K}-1}^{\sys{Q}-1}}\cdot U_\sys{K}^\sys{Q}(r+1)\notag\\
&\overset{(c)}{=}&L_{\sys{K},\sys{Q}}^*(r),\notag
\end{IEEEeqnarray}
where $(a)$ follows \mw{since by  Lemma \ref{lem:UKQ}  the sequence $\{U_\sys{K}^\sys{Q}(u)\}_{u=2}^\sys{K}$ is decreasing and because}  $U_\sys{K}^\sys{Q}(1)=U_\sys{K}^\sys{Q}(2)=C_{\sys{K}-1}^{\sys{Q}-1}$; $(b)$ follows from \eqref{sum:theta_g}; and $(c)$ follows from Theorem \ref{thm:limits} and \eqref{def:UKQ}.
By our assumption $L_{\bm A}=L_{\sys{K},\sys{Q}}^*(r)$, the equality in \eqref{equality:ahold} must hold.
Since $r+1\geq 3$ and the sequence
$\{U_\sys{K}^\sys{Q}(u)\}_{u=2}^\sys{K}$ strictly decreases, equality in
\eqref{equality:ahold} implies that
\begin{IEEEeqnarray}{rCl}
\theta_{g'}=0,\quad \forall\; g'\in[r].\label{S:zero2}
\end{IEEEeqnarray}

Combining \eqref{S:zero} and \eqref{S:zero2}, we conclude that $\bm A$
is a $(r+1)$-regular PDA, and each row has exactly $r$ $``*''$ symbols.
Applying Lemma~\ref{lemma:lowerbound:F}, we complete the proof.

\section{Conclusion}\label{sec:conclusion}
In this work, we have explained
how to convert any Comp-PDA with at least $\sys{K}-\sys{Q}+1$
$``*''$ symbols in each row into a coded computing scheme for a MapReduce
system with $\sys{Q}$ non-straggling nodes.
We have further characterized the optimal storage-communication~(SC)
tradeoff for this system. The Comp-PDA framework allows us to design
\sheng{universal} coded computing schemes with  small file complexities compared to the ones (the MAN-PDAs) achieving the fundamental SC tradeoff.

In our setup, for a given integer storage load $r$, the size of active
set $\sys{Q}$ has to be no less than $\sys{K}-r+1$, since we  exclude
outage events (See Footnote \ref{footnote1}). With a given Comp-PDA, the key to \sheng{obtaining} a coded computing scheme for a given active set is that the subarray formed by the columns corresponding to the active set is still a Comp-PDA. 
\qifa{In fact, for the constructions in P1) and P2), it allows to construct coded computing schemes for some (but not all) active sets if the active set size $\sys{Q}$ satisfies $\lceil\frac{\sys{K}}{r}\rceil\leq\sys{Q}\leq \sys{K}-\sys{r}$.}
\begin{appendices}

\section{Proof of Lemma~\ref{lemma:Z}}\label{sec:claimZ}

We shall prove the first statement of the lemma that the sequence
$\big\{Z_\sys{K}^\sys{Q}(u)\big\}_{u=\sys{K}-\sys{Q}+1}^\sys{K}$ is
strictly convex and decreasing, i.e.,
\begin{IEEEeqnarray}{rCl}
Z_\sys{K}^\sys{Q}(u+1)-Z_\sys{K}^\sys{Q}(u)&<&0,\quad\forall\;u\in[\sys{K}-\sys{Q}+1:\sys{K}-1],\notag\\
Z_\sys{K}^\sys{Q}(u+1)-Z_\sys{K}^\sys{Q}(u)&>&Z_\sys{K}^\sys{Q}(u)-Z_\sys{K}^\sys{Q}(u-1),\quad\forall\;u\in[\sys{K}-\sys{Q}+2:\sys{K}-1].\notag
\end{IEEEeqnarray}
The second statement of the lemma on the piecewise linear function is an immediate consequence of the first one.

By \eqref{eqn:ZKQ},
\begin{IEEEeqnarray}{rCl}
Z_\sys{K}^\sys{Q}(u)
&=&\sum_{l=u+\sys{Q}-\sys{K}}^{\min\{u,\sys{Q}\}}\frac{\sys{Q}-l}{\sys{Q}l}\cdot C_u^l\cdot C_{\sys{K}-u}^{\sys{Q}-l}\notag\\
&=&\sum_{l=u+\sys{Q}-\sys{K}}^{\min\{u,\sys{Q}\}}\frac{C_u^l\cdot C_{\sys{K}-u}^{\sys{Q}-l}}{l}-\sum_{l=u+\sys{Q}-\sys{K}}^{\min\{u,\sys{Q}\}}\frac{C_u^l\cdot C_{\sys{K}-u}^{\sys{Q}-l}}{\sys{Q}}~\IEEEeqnarraynumspace\notag\\
&\overset{(a)}{=}&\sum_{l=u+\sys{Q}-\sys{K}}^{\min\{u,\sys{Q}\}}\frac{C_u^l\cdot C_{\sys{K}-u}^{\sys{Q}-l}}{l}-\frac{ C_\sys{K}^\sys{Q}}{\sys{Q}},\notag
\end{IEEEeqnarray}
where in $(a)$, we used the identity
$\sum_{l=s+\sys{Q}-\sys{K}}^{\min\{u,\sys{Q}\}}C_u^l\cdot
C_{\sys{K}-u}^{\sys{Q}-l}=C_\sys{K}^\sys{Q}$.
Then,
\begin{IEEEeqnarray}{rCl}
&&Z_\sys{K}^\sys{Q}(u+1)-Z_\sys{K}^\sys{Q}(u)\notag\\
&=&\sum_{l=u+1+\sys{Q}-\sys{K}}^{\min\{u+1,\sys{Q}\}}\frac{C_{u+1}^lC_{\sys{K}-u-1}^{\sys{Q}-l}}{l}-\sum_{l=u+\sys{Q}-\sys{K}}^{\min\{u,\sys{Q}\}}\frac{C_u^l C_{\sys{K}-u}^{\sys{Q}-l}}{l}\notag\\
&\overset{(a)}{=}&\sum_{l=u+1+\sys{Q}-\sys{K}}^{\min\{u+1,\sys{Q}\}} \frac{\left(C_u^l+C_u^{l-1}\right)\cdot C_{\sys{K}-u-1}^{\sys{Q}-l}}{l}-\sum_{l=u+\sys{Q}-\sys{K}}^{\min\{u,\sys{Q}\}} \frac{C_{u}^l\cdot\left(C_{\sys{K}-u-1}^{\sys{Q}-l}+C_{\sys{K}-u-1}^{\sys{Q}-l-1}\right)}{l}\IEEEeqnarraynumspace\label{step:sum1}\\
&\overset{(b)}{=}&\sum_{l=\sys{Q}+u+1-\sys{K}}^{\min\{u,\sys{Q}\}}\frac{ C_u^l\cdot C_{\sys{K}-u-1}^{\sys{Q}-l}}{l}+\sum_{l=\sys{Q}+u+1-\sys{K}}^{\min\{u+1,\sys{Q}\}}\frac{ C_u^{l-1}\cdot C_{\sys{K}-u-1}^{\sys{Q}-l}}{l}\notag\\
&&\quad-\sum_{l=\sys{Q}+u+1-\sys{K}}^{\min\{u,\sys{Q}\}}\frac{C_{u}^l\cdot C_{\sys{K}-u-1}^{\sys{Q}-l}}{l}-\sum_{l=\sys{Q}+u-\sys{K}}^{\min\{u,\sys{Q}-1\}}\frac{C_u^l\cdot C_{\sys{K}-u-1}^{\sys{Q}-l-1}}{r}\IEEEeqnarraynumspace\notag\\
&=&\sum_{l=u+1+\sys{Q}-\sys{K}}^{\min\{u+1,\sys{Q}\}}\frac{C_u^{l-1} C_{\sys{K}-u-1}^{\sys{Q}-l}}{l}-\sum_{l=u+\sys{Q}-\sys{K}}^{\min\{u,\sys{Q}-1\}}\frac{C_u^l C_{\sys{K}-u-1}^{\sys{Q}-l-1}}{l}\IEEEeqnarraynumspace\notag\\
&\overset{(c)}{=}&\sum_{l=u+\sys{Q}-\sys{K}}^{\min\{u,\sys{Q}-1\}}\frac{C_u^l C_{\sys{K}-u-1}^{\sys{Q}-l-1}}{l+1}-\sum_{l=u+\sys{Q}-\sys{K}}^{\min\{u,\sys{Q}-1\}}\frac{C_u^l C_{\sys{K}-u-1}^{\sys{Q}-l-1}}{l}\IEEEeqnarraynumspace\notag\\
&=&-\sum_{l=u+\sys{Q}-\sys{K}}^{\min\{u,\sys{Q}-1\}}\frac{C_u^l C_{\sys{K}-u-1}^{\sys{Q}-l-1}}{l(l+1)}\label{step:continue}\\
&<&0,\notag
\end{IEEEeqnarray}
where in $(a)$, we applied the identity
\begin{IEEEeqnarray}{c}
C_{n+1}^{m+1}=C_n^{m+1}+C_n^m;\label{eqn:combination}
\end{IEEEeqnarray}
in $(b)$, we separated the two summations of \eqref{step:sum1} into four summations and eliminated indices of zero terms in the separated summations; and in $(c)$,  we used the change of variable $l'=l-1$  in the first summation.
 Finally, from \eqref{step:continue}, for
 $u\in[\sys{K}-\sys{Q}+2:\sys{K}-1]$, we have
\begin{IEEEeqnarray}{rCl}
&&\left(Z_\sys{K}^\sys{Q}(u+1)-Z_\sys{K}^\sys{Q}(u)\right)-\left(Z_\sys{K}^\sys{Q}(u)-Z_{\sys{K}}^\sys{Q}(u-1)\right)\IEEEeqnarraynumspace\notag\\
&=&\sum_{l=\sys{Q}+u-1-\sys{K}}^{\min\{u,\sys{Q}\}-1}\frac{C_{u-1}^{l}C_{\sys{K}-u}^{\sys{Q}-l-1}}{l(l+1)}-\sum_{l=\sys{Q}+u-\sys{K}}^{\min\{u,\sys{Q}-1\}}\frac{C_{u}^lC_{\sys{K}-u-1}^{\sys{Q}-l-1}}{l(l+1)}\IEEEeqnarraynumspace\notag\\
&\overset{(a)}{=}&\sum_{l=\sys{Q}+u-1-\sys{K}}^{\min\{u,\sys{Q}\}-1}\frac{C_{u-1}^l\cdot \left(C_{\sys{K}-u-1}^{\sys{Q}-l-1}+C_{\sys{K}-u-1}^{\sys{Q}-l-2}\right)}{l(l+1)}-\sum_{l=\sys{Q}+u-\sys{K}}^{\min\{u,\sys{Q}-1\}}\frac{\left(C_{u-1}^l+C_{u-1}^{l-1}\right)\cdot C_{\sys{K}-u-1}^{\sys{Q}-l-1}}{l(l+1)}\label{eqn:sum2}\\
&\overset{(b)}{=}&\sum_{l=\sys{Q}+u-\sys{K}}^{\min\{u-1,\sys{Q}-1\}}\frac{C_{u-1}^l\cdot C_{\sys{K}-u-1}^{\sys{Q}-l-1}}{l(l+1)}+\sum_{l=\sys{Q}+u-1-\sys{K}}^{\min\{u-1,\sys{Q}-2\}}\frac{C_{u-1}^lC_{\sys{K}-u-1}^{\sys{Q}-l-2}}{l(l+1)}\notag\\
&&\quad\quad-\sum_{l=\sys{Q}+u-\sys{K}}^{\min\{u-1,\sys{Q}-1\}}\frac{C_{u-1}^lC_{\sys{K}-u-1}^{\sys{Q}-l-1}}{l(l+1)}-\sum_{l=\sys{Q}+u-\sys{K}}^{\min\{u,\sys{Q}-1\}}\frac{C_{u-1}^{l-1}C_{\sys{K}-u-1}^{\sys{Q}-l-1}}{l(l+1)}\IEEEeqnarraynumspace\notag\\
&=&\sum_{l=\sys{Q}+u-1-\sys{K}}^{\min\{u-1,\sys{Q}-2\}}\frac{C_{u-1}^l C_{\sys{K}-u-1}^{\sys{Q}-l-2}}{l(l+1)}-\sum_{l=\sys{Q}+u-\sys{K}}^{\min\{u,\sys{Q}-1\}}\frac{C_{u-1}^{l-1} C_{\sys{K}-u-1}^{\sys{Q}-l-1}}{l(l+1)}\IEEEeqnarraynumspace\notag\\
&\overset{(c)}{=}&\sum_{l=\sys{Q}+u-\sys{K}}^{\min\{u,\sys{Q}-1\}}\frac{C_{u-1}^{l-1}C_{\sys{K}-u-1}^{\sys{Q}-l-1}}{(l-1)l}-\sum_{l=\sys{Q}+u-\sys{K}}^{\min\{u,\sys{Q}-1\}}\frac{C_{u-1}^{l-1}C_{\sys{K}-u-1}^{\sys{Q}-l-1}}{l(l+1)}\notag\\
&=&\sum_{l=u+\sys{Q}-\sys{K}}^{\min\{u,\sys{Q}-1\}}\frac{2C_{u-1}^{l-1}C_{\sys{K}-u-1}^{\sys{Q}-l-1}}{(l-1)l(l+1)}\notag\\
&>&0,\notag
\end{IEEEeqnarray}
where in $(a)$ we applied the identity \eqref{eqn:combination};
in  $(b)$, we separated the two summations in \eqref{eqn:sum2} and
eliminated the indices of zero terms in the separated summations; and in
$(c)$, we used the change of variable $l'=l+1$. 

\section{Proof of Proposition \ref{proposition}}\label{app:B}

By \eqref{optimal:LQ}, \eqref{LQ} and \eqref{def:UKQ}, we have
\begin{IEEEeqnarray}{rCl}
L_{\sys{K},\sys{Q}}(r)&=&\left(1-\frac{r}{\sys{K}}\right)\cdot\frac{1}{C_{\sys{K}-1}^{\sys{Q}-1}}\cdot U_\sys{K}^\sys{Q}(r),\notag\\
L_{\sys{K},\sys{Q}}^*(r)&=&\left(1-\frac{r}{\sys{K}}\right)\cdot\frac{1}{C_{\sys{K}-1}^{\sys{Q}-1}}\cdot U_\sys{K}^\sys{Q}(r+1).\notag
\end{IEEEeqnarray}
Combining these equalities with \eqref{eq:negative}, we obtain
\begin{IEEEeqnarray}{rCl}
  L_{\sys{K},\sys{Q}}(r)-L_{\sys{K},\sys{Q}}^*(r)&=&-\left(1-\frac{r}{\sys{K}}\right)\cdot\frac{1}{C_{\sys{K}-1}^{\sys{Q}-1}}\cdot\left( U_\sys{K}^\sys{Q}(r+1)- U_\sys{K}^\sys{Q}(r)\right)\notag\\
  &=&\left(1-\frac{r}{\sys{K}}\right)\cdot\frac{1}{C_{\sys{K}-1}^{\sys{Q}-1}}\cdot\sum_{l=\max{\{2,r+\sys{Q}-\sys{K}\}}}^{\min\{r,\sys{Q}-1\}}\frac{ C_{r-1}^{l-1}\cdot C_{\sys{K}-r-1}^{\sys{Q}-l-1}}{l(l-1)}\notag\\
  &\overset{(a)}{=}&\left(1-\frac{r}{\sys{K}}\right)\cdot\frac{1}{C_{\sys{K}-1}^{\sys{Q}-1}}\cdot\frac{1}{r}\cdot\sum_{l=\max{\{2,r+\sys{Q}-\sys{K}\}}}^{\min\{r,\sys{Q}-1\}}\frac{ C_{r}^{l}\cdot C_{\sys{K}-r-1}^{\sys{Q}-l-1}}{l-1},\label{step:m}
\end{IEEEeqnarray}
where in $(a)$, we used the identity $C_{r-1}^{l-1}=\frac{l}{r}\cdot
C_r^l$.
Therefore, with \eqref{optimal:LQ} and \eqref{step:m},
\begin{IEEEeqnarray}{rCl}
\frac{L_{\sys{K},\sys{Q}}(r)-L_{\sys{K},\sys{Q}}^*(r)}{L_{\sys{K},\sys{Q}}^*(r)}
&=&\frac{1}{r}\cdot\frac{\sum_{l=\max{\{2,r+\sys{Q}-\sys{K}\}}}^{\min\{r,\sys{Q}-1\}}\frac{1}{l-1}\cdot
C_{r}^{l}\cdot
C_{\sys{K}-r-1}^{\sys{Q}-l-1}}{\sum_{l=r+\sys{Q}-\sys{K}}^{\min\{r,\sys{Q}-1\}}\frac{1}{l}\cdot
C_{r}^l\cdot C_{\sys{K}-r-1}^{\sys{Q}-l-1}}\notag\\
&\leq&\frac{1}{r}\cdot\frac{\sum_{l=\max{\{2,r+\sys{Q}-\sys{K}\}}}^{\min\{r,\sys{Q}-1\}}\frac{1}{l-1}\cdot C_{r}^{l}\cdot C_{\sys{K}-r-1}^{\sys{Q}-l-1}}{\sum_{l=\max{\{2,r+\sys{Q}-\sys{K}\}}}^{\min\{r,\sys{Q}-1\}}\frac{1}{l}\cdot C_{r}^l\cdot C_{\sys{K}-r-1}^{\sys{Q}-l-1}}\notag\\
&=&\frac{1}{r}\cdot\frac{\sum_{l=\max{\{2,r+\sys{Q}-\sys{K}\}}}^{\min\{r,\sys{Q}-1\}}\frac{l}{l-1}\cdot\frac{1}{l} \cdot C_{r}^{l}\cdot C_{\sys{K}-r-1}^{\sys{Q}-l-1}}{\sum_{l=\max{\{2,r+\sys{Q}-\sys{K}\}}}^{\min\{r,\sys{Q}-1\}}\frac{1}{l}\cdot C_{r}^l\cdot C_{\sys{K}-r-1}^{\sys{Q}-l-1}}\notag\\
&\overset{(a)}{\leq}&\frac{2}{r},\notag
\end{IEEEeqnarray}
where in $(a)$, we used the fact $\frac{l}{l-1}\leq 2$ for any $l\geq 2$.
To prove the second part, we first note
that, by Corollary \ref{cor}, the number of batches required by
constructions P1) and P2) is
 \begin{IEEEeqnarray}{c}
 \sys{F}=\frac{1}{c}\cdot q^{\frac{\sys{K}}{q}}.\label{eqn:Fi}
 \end{IEEEeqnarray}
 On the other hand, to achieve the fundamental SC tradeoff, the number of
 required batches is
 \begin{IEEEeqnarray}{rCl}
 \sys{F}^*&=&C_{\sys{K}}^{r} \notag\\
 &=&\frac{\sys{K}!}{r!(\sys{K}-r)!}\notag\\
 &\overset{(a)}{\geq}&\frac{\sqrt{2\pi }\sys{K}^{\sys{K}+\frac{1}{2}}e^{-\sys{K}}}{e\sqrt{2\pi }r^{r+\frac{1}{2}}e^{-r}\cdot e\sqrt{2\pi}(\sys{K}-r)^{\sys{K}-r+\frac{1}{2}}e^{-(\sys{K}-r)}}\notag\\
 &=&\frac{1}{e^2}\cdot\sqrt{\frac{\sys{K}}{2\pi r(\sys{K}-r)}}\cdot\left(\frac{\sys{K}}{r}\right)^r\left(\frac{\sys{K}}{\sys{K}-r}\right)^{\sys{K}-r}\notag\\
   &=&\frac{1}{e^2}\cdot\frac{q}{\sqrt{2\pi (q-1)\sys{K}}}\cdot q^{\frac{\sys{K}}{q}}\cdot\left(\frac{q}{q-1}\right)^{\sys{K}(1-\frac{1}{q}) },\label{eqn:Fi2}
 \end{IEEEeqnarray}
 where $(a)$ follows by applying  Stirling's approximation $\sqrt{2\pi}n^{n+\frac{1}{2}}e^{-n}\leq n!\leq e\sqrt{2\pi}n^{n+\frac{1}{2}}e^{-n}$ to both  the numerator and the denominator.
 Taking
 the ratio $\frac{\sys{F}}{\sys{F}^*}$ using \eqref{eqn:Fi} and
 \eqref{eqn:Fi2}, we complete the proof of the second part.






\end{appendices}


\begin{thebibliography}{1}
\bibitem{YanISIT2019} Q. Yan, M. Wigger, S. Yang, and X. Tang, ``A
  fundamental storage-communication tradeoff in distributed computing
  with straggling nodes,'' in \emph{Proc. IEEE Int. Symp. Inf. Theory,}  Paris, France, pp. 2803--2807,   Jul. 2019.
\bibitem{MapReduce} J. Dean and S. Ghemawat, ``MapReduce: Simplified
  data processing on large clusters,'' \emph{Sixth USENIX OSDI}, Dec. 2004.


\bibitem{Liu2015} Y. Liu, J. Yang, Y. Huang, L. Xu, S. Li, and M. Qi,
  ``MapReduce based parallel neural networks in enabling large scale
  machine learning,'' \emph{Comput. Intell. Neurosci.}, 2015.


\bibitem{MachineLearing}  C. T. Chu, S. K. Kim, Y. A. Lin, Y. Y. Yu, G.
  Bradski, A. Y. Ng, and K. Olukotun,``Map-reduce for machine learning
  on multicore.'' In \emph{Proc. 20th Ann. Conf. Neural Information Processing Systems (NIPS),} Vancouver, British Columbia, Canada, pp. 281--288, Dec. 2006.

\bibitem{Li2018Tradeoff} S. Li, M. A. Maddah-Ali, Q. Yu, and A. S.
  Avestimehr, ``A fundamental tradeoff between computation and
  communication in distributed computing,'' \emph{IEEE Trans. Inf. Theory,} vol. 64, no. 1, pp. 109--128, Jan. 2018.

\bibitem{Fragouli} Y. H. Ezzeldin, M. Karmoose, and C. Fragouli,
  ``Communication vs distributed computation: An alternative trade-off
  curve,'' in \emph{Proc. IEEE Inf. Theory Workshop (ITW)}, Kaohsiung, Taiwan, Nov. 2017.

\bibitem{YanD3C} Q. Yan, S. Yang, and M. Wigger, ``A
  storage-computation-communication tradeoff for distributed
  computing,'' in \emph{Proc. IEEE Int. Symp. Wire. Commun. System (ISWCS),} Lisbon, Portugal, Aug. 2018.

\bibitem{Yan2018SCC}  Q. Yan, S. Yang, and M. Wigger,
  ``Storage-computation-communication tradeoff in distributed computing:
  Fundamental tradeoff and complexity,'' arXiv:1806.07565.

\bibitem{Qian2017How} Q. Yu, S. Li, M. A. Maddah-Ali, and A. S.
  Avestimehr, ``How to optimally allocate resources for coded
  distributed computing,'' in \emph{Proc. IEEE Int. Conf.  Commun. (ICC), 2017,} Paris, France, 21--25, May. 2017.

\bibitem{Li2017Framework} S. Li, Q. Yu, M. A. Maddah-Ali, A. S.
  Avestimehr,``A scalable framework for wireless distributed
  computing,'' \emph{IEEE/ACM Trans. Netw.,}vol. 25, no. 5, pp. 2643--2653, Oct. 2017.

\bibitem{Li2017Egdge} S. Li, Q. Yu, M. A. Maddah-Ali, and A. S.
  Avestimehr, ``Edge-facilitated wireless distributed computing,'' in \emph{Proc. IEEE Glob. Commun. Conf. (Globlcom), } Washington, DC, USA,  Dec. 2016.

\bibitem{Li2018Wireless} F. Li, J. Chen, and Z. Wang, ``Wireless
  MapReduce distributed computing,'' in \emph{Proc. IEEE Int. Symp. Inf. Theory,} Vail, CO, USA, pp. 1286--1290,   Jun. 2018.

\bibitem{Lampiris2018ISIT} E. Parrinello, E. Lampiris, and P. Elia, ``Coded distributed computing with node
cooperation substantially increases speedup factors,'' in \emph{Proc. IEEE Int. Symp. Inf. Theory,} Vail, CO, USA, pp. 1291--1295,   Jun. 2018.

\bibitem{Song2} S. R. Srinivasavaradhan, L. Song, and C. Fragouli, ``Distributed computing trade-offs
with random connectivity,'' in \emph{Proc. IEEE Int. Symp. Inf. Theory,} Vail, CO, USA, pp. 1281--1285,   Jun. 2018.

\bibitem{Lee2018codes} K. Lee, M. Lam, R. Pedarsani, D. Papailiopoulos,
  and K. Ramchandran, ``Speeding up distributed machine learning using
  codes,'' \emph{IEEE Trans. Inf. Theory.} vol. 64, no. 3, pp. 1514--1529, Mar. 2018.


\bibitem{QYu2017Polynomial}Q. Yu, M. Maddah-Ali, and S. Avestimehr, ``Polynomial codes: an optimal
design for high-dimensional coded matrix multiplication,'' in \emph{Proc. The 31st Annual Conf. Neural Inf. Processing System (NIPS)}, Long Beach, CA, USA, May 2017.
\bibitem{Qyu2018polyminal} Q. Yu, M. Maddah-Ali, and S. Avestimehr,
  ``Straggler mitigation in distributed matrix multiplication:
  Fundamental limits and optimal coding,''  in \emph{Proc. IEEE Int. Symp. Inf. Theory,} Vail, CO, USA, pp. 2022--2026,   Jun. 2018.
\bibitem{LeeMatrix} K. Lee, C. Suh, and K. Ramchandran, ``High-dimensional coded matrix
multiplication,'' in \emph{Proc. IEEE Int. Symp. Inf. Theory,} Aachen, Germany, pp. 2418--2422,   Jun. 2017.
\bibitem{HierarchicalCoding} H. Park, K. Lee, J. Sohn, C. Suh, and J.
  Moon, ``Hierarchical coding for distributed computing,'' in \emph{Proc. IEEE Int. Symp. Inf. Theory,} Vail, CO, USA, pp. 1630--1634,   Jun. 2018.



\bibitem{Cadambe2018} F. Haddadpour, and V. R. Cadambe, ``Codes for
  distributed finite alphabet matrix-vector multiplication,'' in \emph{Proc. IEEE Int. Symp. Inf. Theory,} Vail, CO, USA, pp. 1625--1629,   Jun. 2018.

\bibitem{Kiani2018ISIT} S. Kiani, N. Ferdinand and S. C. Draper,
  ``Exploitation of stragglers in coded computation,'' in \emph{Proc. IEEE Int. Symp. Inf. Theory,} Vail, CO, USA, pp. 1988--1992,   Jun. 2018.


\bibitem{Baharav2018ISIT} T. Baharav, K. Lee, O. Ocal, and K. Ramchandran, ``Straggler-proofing massive-scale distributed matrix
multiplication with $d$-dimensional product codes,'' in \emph{Proc. IEEE Int. Symp. Inf. Theory,} Vail, CO, USA, pp. 1993--1997,   Jun. 2018.



\bibitem{Nuwan2018} N. Ferdinand and S. C. Draper, ``Hierachical coded
  computation,''  in \emph{Proc. IEEE Int. Symp. Inf. Theory,} Vail, CO, USA, pp. 1620--1624,   Jun. 2018.

\bibitem{Reisizadeh2017conf}A. Reisizadeh,
S. Prakash,
R. Pedarsani,
and S. Avestimehr, ``Coded computation over heterogeneous clusters,''  in \emph{Proc. IEEE Int. Symp. Inf. Theory,} Aachen, Germany, pp. 2408--2412,   Jun. 2017.
  \bibitem{secretsharing}R. Bitar, P. Parag, and S. E. Rouayheb, ``Minimizing latency for secure coded computing using secret sharing via staircase codes", \emph{IEEE Trans. Commun.}, early access, 2020.
\bibitem{Tandon2017Gradient} R. Tandon, Q. Lei, A. G. Dimakis, and N.
  Karampatziakis, ``Gradient coding: avoiding stragglers in synchronous
  gradient desent,'' in \emph{Proc. 34th Int. Conf. Machine Learning (ICML),} Sydney, Australia, Aug. 2017.
\bibitem{Raviv2017} N. Raviv, I. Tamo, R. Tandon, and A. G. Dimakis,
  ``Gradient coding from cyclic MDS codes and expander graphs,'' in \emph{Proc.   35th Int. Conf. Machine Learning, (ICML),} Stockholm, Sweden, Jul. 2018.
\bibitem{Zachary2018ISIT} Z. Charles, and D. Papailiopoulos,``Gradient
  coding using the stochastic block model,'' in \emph{Proc. IEEE Int. Symp. Inf. Theory,} Vail, CO, USA, pp. 1998--2002,   Jun. 2018.
\bibitem{Halbawi2018ISIT} W. Halbawi, N. Azizan, F. Salehi, and B.
  Hassibi, ``Improving distributed gradient descent using Reed-Solomon
  codes,'' in \emph{Proc. IEEE Int. Symp. Inf. Theory,} Vail, CO, USA, pp. 2027--2031,   Jun. 2018.



\bibitem{Deniz18} E. Ozfatura,  D. G$\ddot{\mbox{u}}$nd$\ddot{\mbox{u}}$z
and S. Ulukus, ``Speeding up distributed gradient descent
by utilizing non-persistent stragglers,'' in \emph{Proc. IEEE Int. Symp. Inf. Theory,} Paris, France, pp. 2729--2733,   Jul. 2019.

\bibitem{Li2016unified} S. Li, M. A. Maddah-Ali, and A. S. Avestimehr,
  ``A unified coding framework for distributed computing with straggling
  servers,'' in \emph{Proc. IEEE Globecom Workshop,} Washington, DC, USA, pp. 1--6, 2016.
\bibitem{Zhang2018Improved} J. Zhang and O. Simeone, ``Improved
  latency-communication trade-off for map-shuffle-reduce systems with
  stragglers,'' in \emph{Proc. IEEE Int. Conf. Acoust., Speech \& Signal Processing (ICASSP)}, Brighton, UK, May, 2019.

\bibitem{Yan2017PDA} Q. Yan, M. Cheng, X. Tang, and Q. Chen, ``On the
  placement delivery array design for centralized coded caching
  scheme,'' \emph{IEEE Trans. Inf. Theory,} vol. 63,  no. 9, pp. 5821--5833, Sep. 2017.


\bibitem{Maddah2014Fundamental} M. A. Maddah-Ali and U. Niesen,
  ``Fundamental limits of caching,''
\emph{IEEE Trans. Inf. Theory,} vol. 60, no. 5, pp. 2856--2867, May 2014.

\bibitem{d2d2019wang}J. Wang, M. Cheng, Q. Yan, and X. Tang, ``Placement
  delivery array design for coded caching scheme in D2D networks,'' \emph{IEEE Trans. Commun,} vol. 67, no. 5, May 2019.

\bibitem{Yan2018ISIT} Q. Yan, M. Wigger, and S. Yang, ``Placement
  delivery array design for combination networks with edge caching,'' in \emph{Proc. IEEE Int. Symp. Inf. Theory,}  Vail, CO, USA, pp. 1555--1559,   Jun. 2018.


\bibitem{Subpacket2019Private} A. VR, P. Sarvepalli, and A. Thangaraj, ``Subpacketization in coded caching with
demand privacy,'' in \emph{Proc. 26th
National Conf. Commun. (NNC),}Kharagpur, India,  Feb. 2020.

\bibitem{Sun2020medical}R. Sun, H. Zheng, J. Liu, X. Du, and M. Guizani,
  ``Placement delivery array design for the coded caching scheme in
  medical data sharing,'' \emph{Neural Comput. \& Applic. 32}, pp. 867--878, 2020.


\bibitem{Cheng2019}M Cheng, J Jiang, Q Yan, X Tang, ``Constructions of
  coded caching schemes with flexible memory sizes,'' \emph{IEEE Trans. Commun.}, vol. 67, no. 6, Jun. 2019.


\bibitem{Cheng2019variant} M. Cheng, J. Jiang, X. Tang, and Q. Yan,
  ``Some variant of known coded caching schemes with good performance,'' \emph{IEEE Trans. Commun.}, early access, 2019.

\bibitem{Cheng2019group} M. Cheng, J. Jiang, Q. Wang, and Y. Yao, ``A
  generalized grouping scheme in coded caching.'' \emph{IEEE Trans. Commun.}  vol. 67, no. 5, pp. 3422--3430, May. 2019.




\bibitem{HyperGraph} C. Shangguan, Y. Zhang, and G. Ge, ``Centralized
  coded caching schemes: A hypergraph theoretical approach,'' \emph{IEEE Trans. Inf. Theory,}  vol. 64, no. 8,  pp. 5755--5766,  Aug. 2018.

\bibitem{bipartite} Q. Yan, X. Tang, Q. Chen, and M. Cheng, ``Placement delivery array design through strong edge
coloring of bipartite graphs,'' \emph{IEEE Commun. Lett.},  vol. 22, no. 2, pp. 236--239, Feb. 2018.

\bibitem{Tang2018ITW} Q. Yan, X. H. Tang, and Q. Chen, ``Placement
  delivery array and its appllications,'' in \emph{Proc. IEEE Inf. Theory Workshop (ITW)}, Guangzhou, China, Nov. 2018.


\bibitem{linearcodes}L. Tang and A.  Ramamoorthy, ``Coded caching schemes
  with reduced subpacketization from linear block codes,''   \emph{IEEE Trans.  Inf. Theory,}  vol. 64, no. 4, pp. 3099--3120, Apr. 2018.

\bibitem{Unified} K. Shanmugam, A. G. Dimakis, J. Llorca and A. M.
  Tulino, ``A unified Ruzsa-Szemer$\acute{\mbox{e}}$di framework for
finite-length coded caching,'' In proc. \emph{51st Asilomar Conference on Signals, Systems, and Computers,}  Pacific Grove, CA, USA, Oct. 2017.










%
%
%
%
%
%
%
%
%
%
%
%
%
%
%
%
%
%

%
%
%
%
%
%
%
%
%
%

















%
%
%
%


















%
%
%

\end{thebibliography}
\end{document}